\documentclass[a4paper,UKenglish,cleveref, autoref, thm-restate]{lipics-v2021}

\usepackage[skins,breakable]{tcolorbox}

\usepackage[nocompress]{cite}
\nolinenumbers
\newtcolorbox{mybox2}[2][]{%
  attach boxed title to top center
               = {yshift=-8pt},
  fonttitle    = \bfseries,
  title        = #2,#1,
  enhanced,
}
\title{The Cost of Sybils, Credible Commitments, and False-Name Proof Mechanisms} 

\author{Bruno Mazorra Roig}{Universtitat Pompeu Fabra, Spain}{brunomazorra@gmail.com}{https://orcid.org/0000-0003-0779-0765}{}%TODO mandatory, please use full name; only 1 author per \author macro; first two parameters are mandatory, other parameters can be empty. Please provide at least the name of the affiliation and the country. The full address is optional. Use additional curly braces to indicate the correct name splitting when the last name consists of multiple name parts.

\newcommand{\1}{\textbf{1}^T}

\newcommand{\mapname}{\textbf{AnPM}}

\author{Nicol\'as Della Penna}{Amurado Research}{nikete@gmail.com}{}{}

\authorrunning{B. Mazorra and N. Della Penna} %TODO mandatory. First: Use abbreviated first/middle names. Second (only in severe cases): Use first author plus 'et al.'

\Copyright{Bruno Mazorra} %TODO mandatory, please use full first names. LIPIcs license is "CC-BY";  http://creativecommons.org/licenses/by/3.0/

 %TODO mandatory; please add comma-separated list of keywords

\acknowledgements{I want to thank Ethereum Foundation for funding Bruno's research}%optional

\EventEditors{}
\EventNoEds{2}
\EventLongTitle{ATF (CVIT 2023)}
\EventShortTitle{CVIT 2023}
\EventAcronym{CVIT}
\EventYear{2023}
\EventDate{June 17, 2023}
\EventLocation{Little Whinging, United Kingdom}
\EventLogo{}
\SeriesVolume{42}
\ArticleNo{23}
\begin{document}

\maketitle

\begin{abstract}
Consider a mechanism that cannot observe how many players there are directly, but instead must rely on their self-reports to know how many are participating. Suppose the players can create new identities to report to the auctioneer at some cost $c$.  The usual mechanism design paradigm is equivalent to implicitly assuming that $c$ is infinity for all players, while the usual Sybil attacks literature is that it is zero or finite for one player (the attacker) and infinity for everyone else (the 'honest' players). The false-name proof literature largely assumes the cost to be 0. We consider a model with variable costs that unifies these disparate streams. 

A paradigmatic normal form game can be extended into a Sybil game by having the action space by the product of the feasible set of identities to create action where each player chooses how many players to present as in the game and their actions in the original normal form game. A mechanism is (dominant) false-name proof if it is (dominant) incentive-compatible for all the players to self-report as at most one identity. We study mechanisms proposed in the literature motivated by settings where anonymity and self-identification are the norms, and show conditions under which they are not Sybil-proof. We characterize a class of dominant Sybil-proof mechanisms for reward sharing and show that they achieve the efficiency upper bound. We consider the extension when agents can credibly commit to the strategy of their sybils and show how this can break mechanisms that would otherwise be false-name proof.

\end{abstract}

\section{Introduction}
The internet naturally motivates the study of mechanisms where players can use fake names and bids to improve their outcomes. Two literatures, largely separate, have studied this. Using the nomenclature of "Sybil attacks" and focusing on single-player deviations from a generally truthful population \cite{levine2006survey,muller2008sybil,yu2009dsybil,lin2017sybil,chen2019axiomatic}. The other, using the nomenclature of false-name proofness, has focused on the equilibrium of games where players who can miss-represent as multiple ones \cite{yokoo2001robust, suyama2005strategy,todo2011false, alkalay2014false,fioravanti2022false}. In this work, we jointly formulate both streams of the literature and extend the analysis to situations with the potential for commitment. Since the stream of the literature using the term Sybil appears largely unaware of the equilibrium implications of the phenomenon, we use the Sybil term while focusing on equilibrium to help in correcting this. 

The necessary and sufficient condition for a mechanism to be Sybil-proof (equivalently, false-name proof), meaning that players have no incentives to generate Sybils, is that their payoff is no better as they add extra identities is no more than their payoff when they have a single identity.  While several mechanisms have been known not to be Sybil-proof when there is no cost to creating new identities (notably VCG, \cite{yokoo2001robust,sher2012optimal,alkalay2014false}), we generalise the condition to potential costs in the creation of sybils. Motivated by smart contracts, we consider the possibility of credible commitment in the strategies of the Sybil players as a natural extension.  Under these conditions, we demonstrate that many previously studied mechanisms in the literature are not Sybil-proof (\cite{brams1995envy,mcafee1992bidding,mossel2010truthful}), and provide examples of mechanisms that are. 
These are constructed using the \textit{pie shrinking with crowding}. This consists of shrinking the total welfare as more players present themselves, to make the mechanisms Sybil-proof. In this paper, we use this to construct Sybil and truthful for fair-division cake-cutting mechanisms and Sybil-proof bidding rings in second-price auction.

% First example. Trivial game with one action. Players split a reward $R$ pro-rata with number of sybils
\subsection{Simple Example}\label{subsection:example}
% Bound on making this mechanism with no capital (or more generally no way of differencing types).
In many different fields, the Sybil game emerges as a bigger game. One example is when a reward, $R$, is distributed among participants, and the chances of winning the reward are based on the number of identities participating. In this situation, a strategic actor may attempt to create multiple identities to increase their chances of receiving the reward, and so, increase its expected payoff.
More formally, if the expected payoff of any strategic player is $U(x,y) = Rx/(x+y)$ where $x$ is the number of identities reported by the player and $y$ is the number of identities reported by other players. Observe that $U$ is strictly increasing on $x$ and therefore, there is no optimal strategy. Similarly, if all players are strategic, the game has no Nash equilibrium. However, in general, creating fake identities induce some cost to the attacker. If creating an identity has some associated costs $c>0$, then the expected payoff of an attacker creating $x$ identities is $U(x,y) = Rx/(x+y)-cx$. In general, $x=1$ is not a dominant strategy. For example, assume that $R=10$, $c=0.1$, and $y=3$. Then, $U(1,3) = 10/4-0.1=2.4$ and $U(2,3) = 10\cdot 2/5 -0.2=3.8$. And so, players have incentives to report more than one identity. 
\subsection{Related Literature}
In recent years, the security of distributed systems has become increasingly important as more and more of our daily lives are conducted online. One type of attack that has received a lot of attention is the Sybil attack, first identified by John Douceur \cite{douceur2002sybil}. A Sybil attack is a type of attack in which a single malicious entity creates multiple fake identities in order to manipulate the system and gain an unfair advantage.

Sybil attacks have been studied in a variety of contexts, including peer-to-peer networks \cite{dinger2006defending,so2011defending}, online social networks \cite{yu2006sybilguard,yu2008sybillimit}, reputation systems \cite{cheng2006manipulability} blockchain systems \cite{zhang2019double}, combinatorial auctions \cite{yokoo2004effect,alkalay2014false}, and diffusion auctions \cite{chen2022sybil}. In a peer-to-peer network, a Sybil attacker may create multiple fake identities in order to control a large portion of the network and launch a denial-of-service attack. In permissionless anonymous environments, a Sybil attacker may create multiple fake identities in order to spread misinformation or influence public opinion. In a blockchain system, a Sybil attacker may create multiple fake identities in order to control a large portion of the network and carry out a 51\% attack. 
In combinatorial and diffusion auctions, a Sybil attacker may create multiple false identities and bid in different bundles to manipulate the outcome of the auction and increase his gains \cite{yokoo2004effect}.

To protect against Sybil attacks, researchers have proposed a variety of mechanisms, including unique identifier systems \cite{maram2021candid,hashmi2008authentication,sanchez2019zero}, proof-of-work/proof-of-stake systems \cite{baza2020detecting,nakamoto2008bitcoin,king2012ppcoin}, and reputation systems \cite{yu2006sybilguard,kamvar2003eigentrust}. However, these mechanisms are not foolproof, and Sybil attacks can still occur in practice.

In game theory and auction theory, Sybil attacks are usually noted as false-name strategies or shill bids. The first author studying false-name strategies in internet auctions is made in \cite{yokoo2001robust,yokoo2004effect}.
In \cite{yokoo2001robust} the authors present a combinatorial auction protocol that is robust against false-name bids. In \cite{yokoo2004effect}, M. Yokoo et. al. prove that Vickrey–Clarke–Groves (VCG) mechanism, which is strategy-proof and Pareto efficient when there exists no false-name bid, is not false name-proof/Sybil-proof and there exists no false-name proof combinatorial auction protocol that satisfies Pareto efficiency. 
In \cite{iwasaki2010worst}, the authors analyzed the worst-case efficiency ratio of false-name-proof combinatorial auction mechanisms. The authors show that the worst-case efficiency ratio of any false-name-proof mechanism that satisfies some minor assumptions is at most $2/(m+1)$ for auctions with $m$ different goods.
In cite \cite{sher2012optimal}, the author formulates the problem of optimal shill bidding for a bidder who knows the aggregate bid of her opponents. The author proves that if the mechanism is required to be Strategy-proof and false-name proof, then VCG only works with additive valuations.

In the field of non-monetary mechanisms in \cite{todo2011false} the authors study false-name-proof mechanisms in the facility location problem. First, the authors fully characterize the deterministic false-name-proof facility location mechanisms in this basic setting. By utilizing this characterization, they show the tight bounds of the approximation ratios of social cost and maximum cost.

In voting false-name proof mechanisms, different studies have been made \cite{wagman2008optimal,fioravanti2022false}. In \cite{wagman2008optimal}, the authors study voting rules where there is a cost for casting a vote. The authors characterize the optimal (most responsive) false-name-proof with-costs voting rule for 2 alternatives.  They prove that as the voting population grows larger, the probability that this rule selects the majority winner converges to $1$. Also, the authors characterize the optimal group false-name-proof rule for two alternatives.
In \cite{fioravanti2022false} the authors characterize all voting rules that verify false-name-proofness, strategy-proofness, unanimity, anonymity, and neutrality as either the class of voting by quota one (all voters can be decisive for all objects) or the class of voting by full quota (all voters can veto all objects).

To our knowledge, no previous work has been made in Sybil-proof collusion mechanisms behaviour in auctions (and bidding rings), and general fair allocation mechanisms such as cake-cutting with homogenous and heterogeneous valuations.

\subsection{Our contributions}
 The main contributions of this paper can be summarized as follows:
\begin{itemize}
    \item We generalize the games where players can represent more than one identity to the underlying game with some cost that depends on the number of identities, which we call the Sybil extension game. In this game, players can create and utilize Sybils to compete in the underlying normal game. Specifically, we provide a mathematical framework for such games based on the costs associated with creating Sybils and the interactions between a player's Sybils and those of other players. 
    \item We identify a necessary and sufficient condition for a mechanism to be Sybil-proof. That is a necessary condition of a mechanism where players have no incentives to generate Sybils. We use this to demonstrate that many previously proposed mechanisms in the literature do not naturally satisfy this condition, and are therefore not Sybil-proof. 
    \item We analyze collusive behavior in permissionless environments where players want to share some reward $R$ but can create false identities. We call these Reward Distribution Mechanisms (RDM). We find the necessary conditions for an RDM to be a Sybil-proof mechanism and introduce pie shrinking with crowding. Doing so, we find the symmetric, prior-free, budget-independent, and Pareto optimal RDM. We also prove that if players share some knowledge of the number of players, other mechanisms have strictly greater social welfare in equilibrium can be found. More generally, we study Sybil-proof of cake-cutting mechanisms, proving the upper bound of efficiency of Sybil-proof cake-cutting mechanisms and giving a worst-case welfare optimal truthful cake-cutting mechanism.
    \item We also study a more general framework, the collusive behavior in permissionless systems with an unknown number of players participating in bidding rings. We specifically study collusive behavior in the second-price auction, where players can submit false name bids and all players have the same distribution of valuations. In this scenario, we prove there is no efficient Sybil-resistant bidding ring, but there are constructive optimal Sybil-resistant collusion mechanisms.

    \item Motivated by permissionless credible commitment devices, we define Sybil-commitment games. These arise when players can credibly commit to certain strategies. This can be seen as an extended game of a normal game with two phases. The Sybil-commitment phase where players commit a set of rational independent identities with different preferences, and the normal game phase where the Sybils in the first phase have full information about the game. We defined the Sybil-commitment equilibrium and the Sybil-commitment-proof. We proved necessary conditions on the price of anarchy of the underlying game to be Sybil-proof. Finally, we give examples of games that are Sybil-proof but not Sybil-commitment-proof.
\end{itemize}

 \subsection{Organization of the paper}
In this paper, we present a comprehensive study of Sybil-proof mechanisms in mechanism design. The paper is organized as follows:

Section 2 introduces the concept of Bayesian games and games with an unknown number of players. We define the Sybil extension game as a symmetric game with incomplete information and the Sybil Nash equilibrium. We provide examples of non-Sybil proof and Sybil-proof games and discuss the reward distribution mechanism, cake cutting with heterogeneous valuations, and permissionless bidding rings in second-price auctions.

Section 3 introduces  Sybil games with identity commitments. We define the Sybil commitment Nash equilibrium and provide the necessary conditions for a mechanism to be Sybil-proof with identity commitments. We give examples of non-Sybil with commitment-proof games such as Cournot and pro-rata mechanisms.

Finally, in Section 4, we provide conclusions and discuss future research directions.

Overall, the paper provides a comprehensive framework for studying Sybil-proof property in mechanism design, highlights the importance of considering Sybil attacks in the design of mechanisms for distributed systems, and computes the cost of having Sybil-proof mechanisms in terms of efficiency.

\subsection{Notation}

In this paper, we introduce various mathematical notations that are used throughout the text. We start denoting by $\mathcal N$ the set of players, $\Theta_i$ the set of types, by $U_i$ as the utility function of player $i$, which represents the preferences of player $i$ over different outcomes in a game. The action of player $i$ is denoted by $x_i$, while the joint action of all players except player $i$ is represented by $x_{-i}$.
We also use the notation $\textbf{1}$ to represent a vector of $n$ elements with all entries being 1. In the context of group theory, we define $S_n$ as the symmetric group of $n$ elements, which consists of all possible permutations of $n$ elements. Furthermore, $S_\infty$ is the symmetric group of an infinite number of elements.
In the realm of matrices, we use $A^\infty$ to denote the direct sum of infinite copies of a matrix $A$.
Additionally, we introduce the Sybil cost function, denoted by $C$, which measures the cost incurred by a system when facing Sybil attacks or the cost for an attacker to create Sybil identities. Lastly, we use the notation $\text{NE}(G)$ to represent the set of Nash equilibria in a game $G$, where a Nash equilibrium is a stable state in which no player can unilaterally change their strategy and improve their utility.

% General Framework
\section{Sybil Games} 
In this section, we model the Sybil extension game of a game with incomplete information. Informally, the game is modeled as follows. There are a finite but unbounded number of players and two types of players: active players and inactive players. Players types are drawn from a distribution. Each player can choose the number of players to report to the game and a vector of strategies of each Sybil identity to the game. Each Sybil creation induces some costs to the player. The payoff of each player is the sum of the payoffs of the Sybils minus the costs of creation. A game is Sybil-proof if no player has incentives to create false identities.
For the motivating applications of internet-enabled permissionless mechanisms, it appears unlikely for there to be a common prior over types. 
Thus, the natural solution concept for mechanisms is Sybil-proof and dominant strategy incentive compatible (DSIC). 

\subsection{Preliminary}

\textbf{Games with incomplete information, see \cite{roughgarden2010algorithmic}}: In games with incomplete information, each party has a private type $\theta_i \in \Theta_i$, where the joint vector $t = (\theta_1,...,\theta_n)$ is assumed to be drawn from some publicly known distribution $\mathcal D$. The point of such type, $\theta_i$, is that it affects the utility function of party $i$: namely, the utility $U_i$ depends not only on the actions $s_1,...,s_n$, but also on the private type $\theta_i$ of party $i$, or, in even more general games, on the entire type vector $\theta$ of all the parties. With this in mind, generalizing the notion of Nash equilibrium to such games is straightforward. (The resulting Nash equilibrium is also called Bayesian.)

More formally, see \cite{kajii1997robustness} for more details, a game with incomplete information is a tuple $G=(\mathcal N,\Theta, A, p,U)$ consits of:
\begin{enumerate}
    \item A set of players $\mathcal N$.
    \item A set $\Theta = \prod_{i\in\mathcal N}\Theta_i$, where $\Theta_i$ is the (finite) set of possible types for player $i$.
    \item A set of $A=\prod_{i\in\mathcal N}A_i$, where $A_i$ is the set of possible actions for player $i$.
    \item A joint probability distribution $p((\theta_i)_{i\in\mathcal N})$ over types.
    \item Payoff function $U_i:A\times \Theta\rightarrow \mathbb R$ of player $i$.
\end{enumerate}
In this setting, a Bayesian pure strategy for a player $i$ in $G$ is a function $s_i:\theta_i\rightarrow A_i$.

\textbf{Bayesian Nash equilibrium, see \cite{kajii1997robustness}}\footnote{Check also Standford lecture notes \url{https://web.stanford.edu/~jdlevin/Econ\%20203/Bayesian.pdf}}: A Bayesian strategy profile $(s_i)_{i\in\mathcal N}$ is a Bayesian Nash equilibrium if for all $i$,
\begin{equation*}
    \mathbb E_{\theta_{-i}}[u(s_i(\theta_i),s_{-i}(\theta_{-i});\theta_i,\theta_{-i})]\geq \mathbb E_{\theta_{-i}}[u(\tilde{s_i}(\theta_i),s_{-i}(\theta_{-i});\theta_i,\theta_{-i})]\text{ for all }i,\theta_i,\tilde{s_i}.
\end{equation*}
With these definitions, we define a game with an unknown number of players, as the game $G=(\mathcal N,\Theta, A, p,U)$ with incomplete information defined as follows:
\begin{enumerate}
    \item For each player $i$, the set of types is $\Theta_i=\{\text{active,inactive}\}$.
    \item $U_i$ just depends on the action taken by the active players, and $U_i$ is zero for inactive players.
\end{enumerate}
\subsection{Symmetric games with Sybils}

We define the Sybil game of a symmetric normal game \cite{cheng2004notes} given by a function $U$ over a topological space $\textbf{A}$ (in general $\mathbb R^l_+$) of the following form. 
Let $(\textbf{A},0)$ be a pointed topological space and $U:\textbf{A}^\infty\rightarrow\mathbb R^\infty$ continuous map, where $\textbf{A}^\infty$ is the direct sum of infinite numerable copies of $\textbf{A}$\footnote{The direct sum is the set of elements of infinite tuples but just finite number of non-trivial elements, see the appendix for more details.}. We assume that $U_i(x)=0$ for all $x\in\textbf{A}^\infty$ such that $x_i=0$ and we assume that $U$ is anonymous (or symmetric) i.e. $U_{\sigma(i)}(\sigma(x))=U_i(x)$ for all permutations $\sigma \in S_\infty$\footnote{$S_\infty$ denotes the set of bijections of $\mathbb N$.}, meaning that the obtained payoffs are invariant under permutation of different identifiers. We call these maps \textit{anonymous players map} and \textbf{AnPM} for short.

The map $U$ induces a symmetric normal game by restricting to $\textbf{A}^n$ by taking the tuple $G_n=([n],\textbf{A}^n,U_{|\textbf{A}^n})$. Then, for a map of this form, we consider the following static game $\textbf{Sy}(U,C)=(\mathcal D,\tilde{U},\textbf{A}^\infty,C)$:

\begin{enumerate}
    \item Set of players $\mathcal N$ is drawn from a public known distribution $\mathcal D$.
    \item The action space of each (active) player is $\textbf{A}^\infty$. Observe that we can naturally embed $(\textbf{A}^\infty)^n\hookrightarrow \textbf{A}^\infty$. And so, for a vector of action $a\in (\textbf{A}^{\infty})^n$ we have that $U(a)$ is well-defined.
    \item The cost of generating Sybils is modelled by a function $C:\mathbb R_+\times \mathbb R_+\rightarrow\mathbb R_+$ where $C(x,y)$ is the cost of a player that generates $x$ identities and the other players generate $y$ identities. In general, we will assume that $C(x,y)=cx$ for some $c\geq0$ \footnote{However, there are contexts where the cost function can depend on other player's number of Sybils for example if players buy the sybils from a third party that is selling its Sybils. Then the cost per Sybil will follow the laws of supply and demand.}.
    \item Finally, if each player $i$ set of Sybils is $\mathcal I_i$ and if the set of actions taken by all Sybils is $a=(a^1,...,a^N)$ the total payoff of the player $i$ is
    \begin{equation*}
        \tilde{U}_i(a)=\sum_{j\in\mathcal I_i}U_j(a)-C(|\mathcal I_i|,|\bigcup_{k\not=i}\mathcal I_k|).
    \end{equation*}
\end{enumerate}
We call $\textbf{Sy}(U,C)$ the Sybil extension game or Sybil game of $U$. Observe, that the static game $G_N$ is a subgame of $\textbf{Sy}(U,C)$.
In general, we will assume that the cost function is zero or $C(x,y)=cx$ for some positive real value $c$. Observe that if $C(1,y)=0$ and $C(x,y)=+\infty$ for $x\geq2$, then the Sybil game is equivalent to the underlying game since reporting exactly one identity dominates reporting more than one.

In anonymous permissionless environments, the results of players creating Sybils in order to increase their payoff can lead to different negative externalities \cite{levine2006survey,mazorra2022price}. And so, in these domains, it is the desired mechanism that does not lead to Sybil attacks. That is a mechanism such that players have no incentives to Sybil attack. This type of mechanism is called Sybil-proof, Sybil resistant or false-name proof \cite{alkalay2014false,chen2019axiomatic}.

 In the games given by a symmetric map $U$, we say that $U$ is (strictly) \textit{Sybil-proof} if every strategy of the Sybil game $\textbf{Sy}(U,C)$ is (strictly) dominated by a strategy of the subgame $G_N$. More formally, for every player $i$, every set of Sybils $\mathcal I_i$ and tuple of actions $s_i=(a_j)_{j\in \mathcal I_i}$ there exists $\tilde{s}\in \textbf{A}$ such that:
\begin{equation}\label{equation:Sybil_proof}
    U_i(\tilde{s},s_{-i})-C(1,|\bigcup_{k\not=i}\mathcal I_k|)\underset{(>)}{\geq} \tilde{U}_i(s_i,s_{-i})
\end{equation}
Observe that if the game is Sybil-proof for $C=0$ then is Sybil-proof for all $C$.
In case that $(\textbf{A},+,0)$ is a commutative topological monoid (e.g. $(\mathbb R_+,+,0)$), we say that $U$ is Sybil-proof if the equation \ref{equation:Sybil_proof} holds for $\tilde{s} = \sum_{i\in\mathcal I_i} a_i$. That is, for every player $i$, every set of Sybils $\mathcal I_i$ and tuple of actions $s_i=(a_j)_{j\in \mathcal I_i}$ it holds:
\begin{equation}\label{equation:Sybilproof2}
    U_i\left (\sum_{i\in\mathcal I_i} a_i,s_{-i}\right )-C(1,|\bigcup_{k\not=i}\mathcal I_k|)\underset{(>)}{\geq}\tilde{U}_i(s_i,s_{-i})
\end{equation}

As we have seen in \ref{subsection:example}, not all games are Sybil-proof. However, for any $\mapname$ $U$, the Sybil game $\textbf{Sy}(U,C)$ is Sybil-proof. In the following, we will provide a list of functions $U$ that are Sybil-proof for cost function $C=0$:
\begin{enumerate}
    \item The map defined by $U_i(x)=\frac{x_i}{x_i+\sum_{j\not=i}x_j}R-x_i$ for each $i\in\mathbb N$.
    \item With $P,C:\mathbb R_+\rightarrow\mathbb R$ convex function and linear function, respectively, the map $U_i(x)=x_iP(\sum_{i\in\mathbb N}x_i)-C(x_i)$.
    \item All pro-rata games \cite{johnson2022concave}, games of the form $U_i(x_i,x_{-i}) = \frac{x_i}{x_i+x_{-i}}f(x_i+x_{-i})$ with $x_i,x_{-i}\in \mathbb R$ and $f:\mathbb R_+\rightarrow\mathbb R$.
    \item Games of the form $U_i(x,y)= g(x)f(x+y)$ with $g,f:\mathbb R_+\rightarrow \mathbb R$ with $f(0)=0$, concave and $f(w)=0$ for some $w>0$, $g(x)\geq x$, e.g. $g(x)= \max\{x^2,\sqrt{x}\}$.
\end{enumerate}
For a game of this form, for a vector $x$ of actions, we will denote the social welfare as $\textbf{W}(x) = \sum_{i=1}^\infty U_i(x)$. This is well-defined since $x\in\textbf{A}^\infty$.

\textbf{Note}: In some cases, it is interesting to analyse games where players have some constraints (e.g. budget constraints). Assuming that $\textbf{A}=\mathbb R_+$, then we assume that each player has some budget constraint $M_i\in \mathbb R_+$, and so, just can choose actions $a\leq M_i$. In the Sybil game, this extends to the sum of actions taken by its Sybils can not exceed $M_i$, i.e., $\sum_{j\in\mathcal I_i}a^i_j\leq M_i$. We will call this game Sybil game with budget constraints.

% Hedge fund
\subsection{Examples}
Let's consider the following normal game with $n$ players, where a reward $R$ is shared among the identities reported to the mechanism. Each active player can report an arbitrary number of identities, and each creating of identity has a cost $c\geq 0$. Each identity has a trivial set of actions (participate). And so, the Sybil extension game can be modelled as a game with unknown number of players and payoff:
 \begin{equation}\label{equation:payoffexample}
     U_i(x,y) = \frac{x}{x+y}R-cx,
 \end{equation}
 where $x$ is the number of players reported by the $i$th player and $y$ is the total number of identities reported by the other players. In this scenario, as we show in the following proposition, the welfare loss is of the order of magnitude of $n$. 

\begin{proposition}\label{prop:game1} The game defined by \ref{equation:payoffexample} has a symmetric mixed Nash equilibrium. The strategy consists of randomizing between two actions. The mixed strategy is defined as
\begin{equation*}
    \pi=\begin{cases}
        \lfloor \frac{R/c(n-1)}{n^2}\rfloor,\text{ with probability }p,\\
        \lceil \frac{R/c(n-1)}{n^2}\rceil,\text{ with probability }1-p,
        \end{cases}
\end{equation*}
for some $p\in [0,1]$. The welfare in equilibrium is $\Theta (R/n)$ and so the ratio of the optimal social welfare and the welfare in equilibrium is $\Theta(n)$.
 \end{proposition}

In other words, there is a loss of welfare of order of $\Theta(n)$ due to the ability of players to submit false identities. We will see this phenomenon is more general in the following section. We call this ratio the \textit{Price of Identity loose}.

\subsection{Symmetric mechanisms with Sybils}
In the previous section, we have defined the Sybil extension of a game where each player has a symmetric payoff functions in terms of the outcomes. In this section, we will discuss Symmetric mechanisms where players do not necessarily have Symmetric payoff on these outcomes.

A mechanism for finite but bounded number of players is given by:
\begin{enumerate}
    \item Players types spaces $T_i$ and players action spaces $A_i$ for all $i\in\mathbb N$. We assume that there exist a null action $0\in A_i$ (not acting or not participating in the mechanism).
    \item An alternative set $X$.
    \item Players valuations functions $V_i:T_i\times X\rightarrow \mathbb R$.
    \item An outcome function $a:\bigoplus_{j=1}^\infty A_i\rightarrow X$.
    \item Payment functions $p_i:\bigoplus_{j=1}^\infty A_i\rightarrow \mathbb R$.
\end{enumerate}
The game with strict incomplete information induced by the mechanism is given by using the types spaces $T_i$, the action spaces $X_i$ and the utility functions $U_i(t_i,x)=V_i(t_i,a(x))-p_i(x)$. We say that a mechanism is symmetric if:
\begin{itemize}
    \item $A_i = A_j$ for all $i,j$, so we can write $A = A_i$
    \item $a$ is invariant under permutations, i.e.,  holds $a(\sigma(x))=a(x)$ for all $x\in\bigoplus_{j=1}^\infty A_j$ and $\sigma\in S_\infty$
    \item For all permutations $\tau\in S_{\infty}$ such that $\tau(i)=j$ and $\tau(j)=i$, it holds $p_i(\tau(x))=p_j(x)$ for all $x\in\bigoplus_{j=1}^\infty A_j$.
\end{itemize}
Less formally, a mechanism is symmetric if the outcome does not depend on the identity taking the action just the joint set of actions. 

A symmetric mechanism is \textit{Sybil-proof} if no players have incentives to report more than one action. That is, we can consider the Sybil extension game, where players can report more than one action, so is Sybil proof if for every finite set $\mathcal I\subseteq \mathbb N$ (w.l.o.g. we can assume $\mathcal I=[k]$) and every action tuple of actions $x$ exist $\tilde{x}=(z,0,...,0,x_{-[k]})$ such that
\begin{equation*}
    U_1(\tilde{x})\geq \sum_{i=1}^k U_i(x)
\end{equation*}
Examples of Sybil-proof mechanisms are First-price auction and second-price auction. On the other hand, cake-cutting mechanisms, cost-sharing mechanisms and  VCG combinatorial auctions are in general not Sybil-proof (more details in following sections).

\subsection{Reward Distribution Mechanisms}

The distribution of resources among multiple players is a vital concern in both economics and game theory. One key and straightforward category of challenges in this area involves creating mechanisms for distributing rewards. These mechanisms seek to allocate a divisible asset with a symmetric value of $R$ to players who have a linear payoff function. The main goal of these mechanisms is to maximize the social welfare.
However, efficient mechanisms are often vulnerable to false-name or Sybil strategies, in which a single player pretends to be multiple identities in order to gain a larger share of the resources. 
In this section, we propose a dominant strategy incentive compatible (DSIC) and Sybil-proof mechanism for allocating a part of the fungible item. We also prove that this mechanism is the optimal solution for this class of mechanisms. 
The mechanism will consist of sufficiently shrinking the "pie" per self-reported identity. We call this technique the \textit{pie shrinking with crowding}. Furthermore, we investigate the case where players have some information about the number of players, and explore how this impacts the design and performance of the mechanism.

 To investigate the issue of Sybil resistance in the context of reward distribution mechanisms, we assume that some notion of value is shared among players. Moreover, we assume that the cost of generating Sybils is $c=0$ (in general, results will also apply when $c>0$). 
 The objective of the reward disturbing mechanism is to distribute some reward $R$ among players. 
 A reward distribution mechanism will consist of a $\mapname$ function $U:\textbf{A}^\infty_+\rightarrow \mathbb R^\infty$ (with $\textbf{A}=\mathbb Z_{\geq0}\text{ or }\mathbb R_+$) so that is symmetric, aggregative and $\sum_{i=1}^\infty U_i(a)\leq R$ for all $a\in \textbf{A}^\infty$. 
 We denote by $\mathcal M(\textbf{A})$ the set of mechanisms of this form.

We say that a mechanism $U\in\mathcal M(\textbf{A})$ is \textit{symmetric in equilibrium} if $U$ is symmetric and the equilibriums of the normal game $G_n=([n],U_{\mid \textbf{A}^n},\textbf{A}^n)$ have the same payoff for all players. That is, for all $x\in \text{NE}(G_n)$, we have that $U_i(x)=U_j(x)$ for all $i,j\in[n]$. We say that the mechanism $U\in\mathcal M(\textbf{A})$ is Sybil-proof if the equation \ref{equation:Sybilproof2} holds.

One approach to addressing the issue of Sybil-proofness in reward distribution mechanisms is to consider mechanisms that assume the existence of some notion of shared value, such as capital or money, among the players. Under this assumption, the mechanism can use the value locked or burned to determine the reward of the player. To describe optimal allocation mechanisms in equilibrium, we consider the \textit{pro-rata} games, a family of mechanisms that are Sybil-proof and symmetric.

The pro-rata games defined in \cite{johnson2022concave} are a particular type of aggregative games \cite{jensen2018aggregative}\footnote{Actually, every aggregative, continuos, Sybil and collusion resistant is a pro-rata game}. A pro-rata game with $n$ players is defined as the game with the following payoff function for players $i=1,...,n$:
\begin{equation}\label{payoffpro}
    U_i(x) = \frac{x_i}{\1\cdot x}f(\1\cdot x).
\end{equation}
Here, $f:\mathbb R_+\rightarrow\mathbb R$ is some function satisfying $f(0)\geq0$, while $x\in\mathbb R^n_+$ is a non-negative vector whose $i$th entry is the action performed by the $i$th player. The games such that $f(0)=0$ and $f$ is concave are called concave pro-rata games. We define the discrete pro-rata games as the pro-rata game restricted to $\mathbb Z_{\geq0}$.

Given the pro-rata games, if we want to allocate the reward $R$ we can choose a function $f$ such that $\max_{z\geq0}f(z)=R$, and we allocate to each player the amount $U^f_i(x_i,x_{-i})$ that is the payoff of the pro-rata game with function $f$, see equation \ref{payoffpro}. 

An initial example involves the constant function $f(x)=R$, leading to the pro rata game related to $f$ being $U_i(x_i,x_{-i}) = \frac{x_i}{x_i+x_{-i}}R$. Note that for any set of actions taken by the players, this results in an optimal welfare distribution. It is evident that all players aim to present (or lock) the highest possible value\footnote{This mechanism is typically referred to as proof of stake (PoS) \cite{king2012ppcoin}.}. In situations where all players possess a finite amount of value, a Nash equilibrium exists, since locking the entire budget is a strictly dominant strategy.
And so this mechanism is prior-independent. Otherwise, if players have no budget constraints, then the game does not have equilibrium. It is important to note that this mechanism, while being Sybil-proof and welfare optimal, has several drawbacks that limit its practical usefulness. In particular, the mechanism tends to disproportionately reward more capitalized players, leading to an asymmetric distribution of rewards. 
Additionally, the mechanism is unrealistic in the sense that it does not account for the fact that players may be able to borrow money in order to increase their payoffs. In reality, loans would come with fees that would need to be taken into consideration in the design of the mechanism. As a result, it may be necessary to consider alternative mechanisms that address these limitations in order to achieve a more symmetric and realistic reward distribution. Moreover, capitalized players would have some additional loss due to the inflation rate.

Now, we will try to find the optimal social welfare reward distribution mechanism, symmetric (the mechanism is independent of players identifiers), Sybil-proof, with strictly dominant strategy and budget independent. The need for symmetry on the mechanism comes from the arbitrariness (or the intrinsic characteristics) of the order of the identifiers. 
\begin{lemma}\label{lemma} Any symmetric, strategy-proof and Sybil-proof anonymous map mechanism $U$ (all players have the same valuation and type) can be reformulated as a discrete pro-rata mechanism with action space being the number of identities reported.
\end{lemma}

By this lemma, is enough to study the mechanism of the following form.
Let $r:\mathbb N\rightarrow \mathbb R$ be a decreasing map bounded by $r(1)=:R$. Assume that $n$ players are reported to the mechanism, then $r(n)$ is split among the reported players. So, the total payoff distributed is $r(n)$ and each player obtains $r(n)/n$ in equilibrium. Clearly, if $r(n)=0$ for all $n\geq 3$ this mechanism is Sybil-Proof, however, the social welfare is zero. Another example of mechanism is $r(n)=R$, however, as mentioned before, this mechanism is not Sybil-proof in general.
 
These types of mechanisms are completely determined by the function $r(n)$. Previously, we presented a mechanism that is Strategy-proof and Sybil-Proof, but is not independent of players budget constraints. In the following, we will give the optimal mechanism in the family of mechanisms of a particular form. That is, given the family of mechanisms 
\begin{equation*}
    \mathcal M=\{r:\mathbb N\rightarrow [0,R]\mid r\text{ induces a Sybil-Proof mechanism }\forall c>0,n\},
\end{equation*}
the optimal social welfare mechanism is $r_{max}(n)=\text{argmax}_{r\in\mathcal M}\{r(n)\}$.

Observe that if $l$ Sybils are reported to the mechanism, the payoff obtained by a Sybil is $r(l)/l$ and so, the payoff of the player is $U_i(x,y) = x/(x+y)r(x+y)-cx$. Therefore, is a discrete pro-rata game. Now take a player $i$ and assume that there are $y$ of reported players in the game. Then, the best response is $x\in\mathbb Z_{\geq0}$ such that

\begin{equation*}
\begin{aligned}
& \underset{x}{\text{maximize }}x\frac{r(y+x)}{y+x}-cx\\
& \text{subject to } x\in\mathbb Z_{\geq0}
\end{aligned}
\end{equation*}

Therefore, a mechanism $r\in \mathcal M$, if and only if, $\frac{r(1+y)}{1+y}\geq x\frac{r(y+x)}{y+x}$ for all $x\geq1,y\geq0$. In particular, $r$ holds
$r(1+y)/(1+y)\geq 2r(2+y)/(2+y)$ for all $y\geq0$. With this inequality, it can be deduced the following proposition.
\begin{proposition}\label{prop:optimalproof} The distribution mechanism given by $r(n)=\frac{n}{2^{n-1}}R$ is a Sybil-proof mechanism for all $c\geq0$. Moreover, $r=r_{max}$ i.e. the mechanism proposed, is the social welfare optimal in $\mathcal M$.
%In particular, we have that the cost of Sybils of the reward distribution mechanism game is $\textbf{CoS}=\frac{n}{2^{n-1}}$.
\end{proposition}

\begin{figure}[!h]
    \centering
    \includegraphics[scale=0.275]{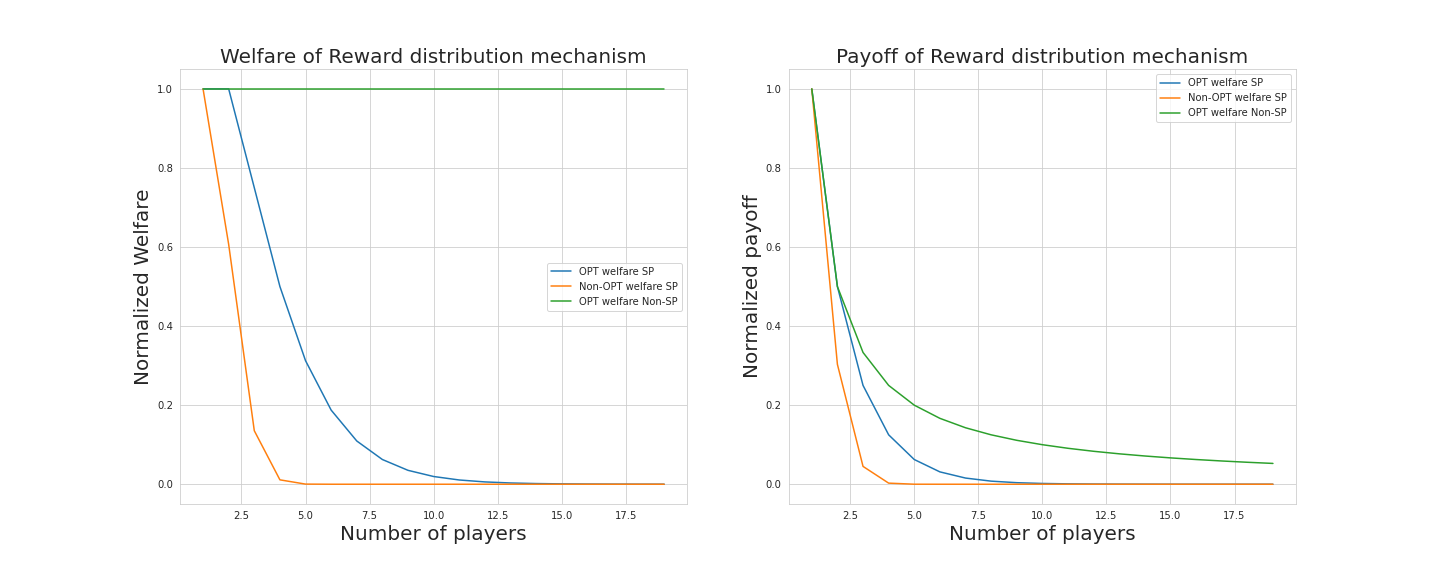}
    \caption{Welfare and payoff of reward distribution mechanisms}
    \label{fig:my_label}
\end{figure}

Observe that with the reward mechanism $r_{max}$, is by construction, strictly dominating strategy to self-report just once. Also, an implication of the proposition is that there does not exist an efficient prior-free reward mechanism. Informally, these mechanisms consists of shrinking the pie optimally to have an optimal Sybil-proof reward mechanism. This technique can also be used if the item is non-divisible.

If the item is non-divisible and players are risk-neutral and share the same valuation on the item, then a similar Sybil-proof mechanism can be achieved by random allocating the item to any player with probability $\Pr[\text{item allocated to the $i$th player}]=\frac{1}{2^{n-1}}$. Clearly, there is (large) inefficiency lost since with probability $1-\frac{n}{2^{n-1}}$ the item is not allocated to any player.

Now we will discuss if players have some knowledge on the number of players. We begin by demonstrating the existence of a family of mechanisms that achieve convergence to the optimal allocation in equilibrium, when the players participating know the total number of players and show that this equilibrium is unique. However, we also highlight the fact that these mechanisms are dependent on the information shared among players, and thus are not robust. 

More formally, we will discuss the following problem. Try to find a mechanism that maximizes the worst-case equilibria with an unknown number of players. That is, the distribution $\mathcal D$ where the number of players is drawn is unknown for the mechanism designer but known for the players. In this case, the players play a Bayesian game with unknown number of players\footnote{In the appendix we discuss the prior-optimal problem, i.e. the problem where the mechanism designer also knows $\mathcal D$.}. We define  $W^{-}(r,\mathcal D) = \min_{x\in\text{NE}(\mathcal D)}\{\mathbb E[\sum_i r_i(x)]\}$ as the worst expected social welfare in Bayesian equilibrium with distribution $\mathcal D$ and mechanism $r$.
\begin{equation}\label{eq:optimal_mechanism}
    r_{\text{max}} = \text{argmax}_{r\in\mathcal M}\min_{\mathcal D} W^{-}(r,\mathcal D)
\end{equation}

An example of a mechanism such that does not reward more capitalized players is the pro-rata mechanism associated to:
\begin{equation*}
    f(x) = \begin{cases} 
                \frac{R}{K}x, \text{ if }0\leq x\leq K,\\
                0,\text{ otherwise.}
            \end{cases}
\end{equation*}
The pro-rata mechanism $U^f$ with known number of players have as set of equilibriums $\{x: \1\cdot x=K\}\subseteq\text{NE}$. And therefore $U^f$ have some equilibrium that achieve OPTW. However, the set of equilibriums contain asymmetric equilibriums, i.e. points such that different players obtain different payoffs. Moreover, there are some equilibriums such that the payoff welfare is zero, e.g. all report $K$. This motivates the following definition.
\begin{definition} A reward distribution mechanism is said to be $K$-\textit{symmetric} if, when the game is played among players with more than $K$ units of value, all players receive the same payoff in equilibrium. If a reward distribution mechanism is $K$-symmetric for all $K$, we say that is budget independent.
\end{definition}
An example of $R$-symmetric pro-rata mechanism is the one given by $f(x) = R-x$. This function, induces the payoff function $U_i(x_i,x_{-i}) = \frac{x_i}{x_i+x_{-i}}R-x_i$. Observe that in this mechanism, using more than $R$ units of value is strictly dominated strategy. Since $f$ is linear, is in particular concave, and therefore, by \cite{johnson2022concave}, if the number of players is known, we have that this game has a unique Nash equilibrium and is of the form $x=(q/n)\textbf{1}$, where, $q$ is the solution of $(n-1)f(q)+qf'(q)=0$ and $n$ is the number of players. And so, the unique equilibrium is $x = R\frac{n-1}{n^2}\textbf{1}$ with payoff $U_i(x) = \frac{R}{n^2}$ for $i=1,...,n$. Then, in equilibrium, the pro-rata mechanism has social welfare $R/n$, so is not an optimal welfare reward distribution mechanism. 

If the players know the total number of players $n$ with no uncertainty, then worst case equilibrium optimality can be achieved asymptotically.
\begin{proposition}\label{prop:nearoptimal} There is a family of functions $\{f_k\}_k$ such that the pro-rata mechanisms $\{U^{f_k}\}_k$ are $R$-symmetric and the social welfare in equilibrium converge to the reward $R$.
\end{proposition}

Therefore, for a reward $R$ if the number of players is information shared among players, there exists a symmetric, Sybil-proof, $R$-symmetric mechanism that achieves nearly optimal welfare. However, the optimality of the mechanism depends on the number of active players being public information among  the players. This, in general, is not true in internet and blockchain based applications. And so, a natural solution concept is the strategy and Sybil-proof pro-rata mechanisms. However, we will see that these mechanisms do not do better, in terms of social welfare,  than the discrete pro-rata mechanisms.

\textbf{DSIC and  Sybil-proof pro-rata mechanism}: Now, let's assume that all players have at least $K$ units of value. We will find constructively a function $f$ such that the pro-rata mechanism has a strictly dominant strategy and $\max f = R$. We will find this function constructively. Let's assume that we want $x=R$ to be the strictly dominant strategy. Therefore, for all $y$ we have that $x=R$ is a local maximum of $U_i(x,y)=\frac{x}{x+y}f(x+y)$ for all $y$. Therefore, it holds that  $\frac{\partial U_i}{\partial x}(K,y) = 0$. So,
\begin{equation*}
    \frac{y}{(K+y)^2}f(K+y)+\frac{K}{K+y}f'(K+y)= 0,\text{ for all }y>0.
\end{equation*}
The unique solution of the last differential equation such that $\text{max}_{x\geq0} f(x)=R$ is $f(x) = \frac{Re}{K}xe^{-x/K}$.  This pro-rata mechanism, has a unique Nash equilibrium and the social welfare with $n$ players is $f(nK)= R\cdot n\cdot e^{-n+1}$. In summary, we have proved the following proposition:
\begin{proposition} There is a unique pro-rata mechanism with twice differentiable function $f$ with strictly dominant strategy that distributes a reward $R$. The mechanism has welfare in equilibrium $R\cdot n\cdot e^{-n+1}$. In particular, we have that  $\text{sup}_{r\in\mathcal M}\min_{\mathcal D} W^{-}(r,\mathcal D)\geq R\cdot n\cdot e^{-n+1}$.
\end{proposition}

\textbf{Sybil-proof prior-free truthful cake cutting mechanism.}

A natural generalization of the reward distribution problem is the cake-cutting problem, where a divisible good is to be divided among multiple players with heterogeneous valuations. This problem has been well-studied in the literature \cite{brams1995envy,mossel2010truthful}. Previously, we presented a specific cake-cutting mechanism for the case of divisible goods with homogeneous valuations. However, many real-world scenarios involve a heterogeneous good, such as a cake with different flavors or a public good with different features. In this section, we will provide a more general framework for the cake-cutting problem that can handle such heterogeneous goods. We will build upon the work of \cite{mossel2010truthful} to present a modified mechanism that incorporates Sybil resistance, providing a solution for this important problem. We will provide a detailed description of the mechanism and its underlying principles, as well as an analysis of its performance. To do so, we will define cake-cutting mechanisms for finite but unbounded number of players. We will reformulate different notions of fairness, such as envy-freeness and proportionality. We will prove that no Sybil-proof mechanism is efficient, and we will upper bound the worst-case social welfare. To do so, we will again use pie shrinking with crowding.

\textbf{Cake-cutting model}: Let $\mathfrak C$ be a $\sigma-$algebra of the set $\mathcal C$ ($\mathfrak C$ model the slices of the cake, and $\mathcal{C}$ the cake). Let there be $n$ players and let $\mu_1,...,\mu_k$ be non-atomic probability measures on $(\mathcal C,\mathfrak C)$, so the value of the slice $C\in\mathfrak C$ is $\mu_i(C)$. Let $\mathcal L=\{\mu:\text{non-atomic probability measure of }(\mathcal C,\mathfrak C)\}\cup\{0\}$, where $0$ is the map that sends everything to zero. A cake-cutting mechanism with finite but unknown bounded number of players consists of a (random) allocation rule\footnote{ here we denote by $\mathcal L^\infty =\{v\in\prod^\infty_{i=1}\mathcal L:\text{ with finite number of non-trivial components}\}$, and similar for $\mathcal C^\infty$.} $\mathfrak A:\mathcal L^\infty\times\Omega\rightarrow \mathfrak C^\infty$ an where if players $i=1,...,k$ report to the mechanism the valuations $\mu_1,...,\mu_k$ and an event $w\in\Omega$ is drawn, then 
\begin{equation*}
    \mathfrak A((\mu_1,...,\mu_k,0,...),w)_j = \begin{cases}
    \text{slice allocated to player }i,\text{ if }1\leq j\leq k\\ \emptyset,\text{ otherwise}\\
    \end{cases}
\end{equation*}
and for all $(v,w)\in\mathcal L^\infty\times \Omega$ it holds $\mathfrak A(v,w)_i\cap \mathfrak A(v,w)_j=\emptyset$ for all $i,j$. Moreover, we will assume that the allocation is symmetric. That is, for all $(v,w)\in\mathcal L^\infty\times \Omega$ and all permutation $\sigma\in S^\infty$ it holds $\mathfrak A(v,w)_i = \mathfrak A(\sigma(v),w)_{\sigma(i)}$. In cake-cutting literature \cite{robertson1998cake}, different fairness criteria are taken into account for designing cake-cutting mechanisms. In the following, we reformulate the most important once in cake-cutting mechanisms with a finite (but not bounded) number of players.

A mechanism $\mathfrak A$ is:
\begin{itemize}
    \item \textbf{Envy-free in expectation} if, for every allocation, no player prefers another player's slice of cake. More formally, if for all vector of valuations $\mu\in \mathcal L^\infty$ and all $i,j$ it holds $\mathbb E[\mu_i(\mathfrak A(\mu)_i)]\geq\mathbb E[ \mu_i(\mathfrak A(\mu)_j)]$.
    \item \textbf{$\alpha$-proportional} if for all vector of valuations $\mu\in\mathcal L^\infty$ and all $i$ such that $\mu_i\not=0$ it holds $\mu_i(\mathfrak A(\mu,w)_i)\geq \alpha$. We say that is in expectancy $\alpha$-proportional if the same holds in expectancy.
    \item \textbf{Non-wastefulness} if for all vector of valuations $\mu\in\mathcal L^\infty$ and event $w\in\Omega$ we have that $\bigcup_{i=1}^\infty \mathfrak A(\mu,w)_i = C$.
    \item \textbf{Truthful in expectation} if truth-telling is a weak-dominant strategy.  More formally, it holds $E[\mu(\mathfrak A(\mu,v_{-i})_i)]\geq\mathbb E[\mu(\mathfrak A(\nu,v_{-i})_i)]$ for all $\mu,\nu\in\mathcal L,v\in\mathcal L^\infty$ and $i\in\mathbb Z_{\geq0}$.
\end{itemize}
When extending the cake-cutting mechanism to include the possibility of Sybil/false-name valuations, where each player can report multiple (but finite) identities to the mechanism, traditional notions of fairness such as proportionality and envy-freeness no longer apply. For example, in a cake-cutting game where one player uses a Sybil attack to submit multiple valuations, a proportional cake-cutting mechanism would result in an asymmetric distribution as the attacker would receive a larger portion of the cake. Similarly, envy-freeness, which ensures that no player prefers another player's portion of the cake, cannot be guaranteed as players do not know which identities reported to the mechanism belong to each player. A natural way of keeping these properties in the Sybil extension game is by making the mechanisms Sybil-proof. If players have no incentives to report more than one valuation, then all mathematical definitions stated before maintain their intrinsic characteristics. Using the definition of Sybil-proofness, we have that a cake-cutting mechanism $\mathfrak A$ is Sybil-proof if holds:
\begin{equation}
    \mathbb E[\mu_i(\mathfrak A(\mu_i,\mu_{-i},\cdot)_i)]\geq \mathbb E[\sum_{j=1}^k\mu_i(\mathfrak A(\nu_{i_1},...,\nu_{i_k},\mu_{-i},\cdot)_{i_j})]\text{ for all }\mu_i,\nu_{i_1},...,\nu_{i_k}\in \mathcal L,\mu_{-i}\in\mathcal L^\infty.
\end{equation}
In the following proposition, we will see that Sybil-proofness is a strong condition. This condition implies that the mechanism waste (a large) part of the cake if other properties such as proportionality want to be guaranteed. Moreover, we upper bound the worst-case social welfare of Sybi-proof cake-cutting mechanisms.
\begin{proposition}\label{prop:cake1} No truthful Sybil-proof cake cutting is $\alpha$-proportional for $\alpha >1/2^{n-1}$. In particular, truthful Sybil-proof cake cutting are not non-wasteful. So the worst-case social welfare is upper bounded by $n/2^{n-1}$, i.e. $\underset{\mu\in\mathcal L^n}{\min}\mathbb E[\sum_{i=1}^n\mu_i(\mathfrak A(\mu,\cdot)_i)]\leq n/2^{n-1} $.
\end{proposition}
In the previous proposition, we have seen that all mechanisms hold $\underset{\mu\in\mathcal L^\infty}{\min}\mathbb E[\sum_{i=1}^\infty\mu_i(\mathfrak A(\mu,\cdot)_i)]\leq n/2^{n-1} $. In the following, we will define a (non-constructive) mechanism that is truthful, Sybil-proof, and is $1/2^{n-1}-$proportional. In order to do so, we will use Neyman's theorem (similar to \cite{mossel2010truthful}) that establishes that there exists a partition $C_1,...,C_k$ of the cake $\mathfrak C$ such that for all players $i$ and slices $j$ it holds that $\mu_i(C_j)=1/k$. The mechanisms works as follows:

Assume that players' true valuations are $\mu_1,...,\mu_n$ and that each one declares some measure $\nu_i$. First, find a partition $C_1,...,C_n$ such that for all $i,j$ holds $\nu_i(C_j)=1/n$. Then choose a random permutation $\sigma\in S_n$ from the uniform distribution. Afterward, toss a biased coin $X$ with probability $\Pr[X=1]=n/2^{n-1}$. If $X=1$, then allocate $C_{\sigma(i)}$ to the $i$th player, otherwise allocate the empty set.

\begin{proposition}\label{prop:cake2} The mechanism Sybil-proof (non-constructive) cake-cutting mechanism that is in expectancy $1/2^{n-1}-$proportional.
\end{proposition}

Putting it all together, we have proved that 
\begin{equation*}
    \max_{\mathfrak A\in\mathcal M}\min_{\mu\in\mathcal L^n}\mathbb E[\sum_{i=1}^n\mu_i(\mathfrak A(\mu,\cdot)_i)] = n/2^{n-1}
\end{equation*}
where $\mathcal M$ is the set of all truthful Sybil-proof cake-cutting mechanisms.

% Permissionless bidding ring.

\textbf{Bidding ring in second price auction with an unknown number of players}.

Another generalization of reward sharing  arises in mechanisms such that players want to collude in order to reduce competition and increase their expected profits. For example, in auctions, the bidders can act collusively and engage in a collusion with a view to obtaining lower prices, see \cite{krishna2009auction}. The resulting arrangement, usually nominated by \textit{bidding ring}, is studied in \cite{mcafee1992bidding,marshall2014economics,marshall2007bidder} in the different auctions and different conditions in the properties of the cartel.

In a second-price auction, bidders submit sealed bids and the highest bidder wins the auction, but pays the price of the second-highest bid rather than their own. If there are ties, item is allocated randomly among the identities that made the highest bid. Clearly, second price auctions (and standard auctions with one item) are Sybil-proof (in case that the valuation of the players have no mass points). For example, in a first price auction, if two or highest bids coincide, then players have incentives to bid more than once in order to increase their chances to obtain the item. However, if the distribution has no mass points, no player can benefit from bidding more than once. We will discuss mechanisms where players try to collude. Bidding rings can occur in this type of auction when a group of bidders colludes to manipulate the outcome of the auction. To illustrate the usefulness of the pie-shrinking technique, we will consider the model proposed in \cite{mcafee1992bidding}. There are $n$ risk-neutral bidders, denoted by $i=1,...,n$, each player $i$ knows his own valuation $v_i$ of the item, while all other bidders and the seller perceive $v_i$ as an independent draw from a cumulative distribution $F$. The seller sells the item through a second-price auction. If two or more players have the highest bid, then the item is allocated randomly. Assume that $F$ has differentiable density $f$ with support $[0,v_h]$ and $F$ is common knowledge. Moreover, neither the ring center nor the players know the number of players $n$. In \cite{mcafee1992bidding}, the authors first introduce a mechanism with no side payments. They prove that if the function $H = (1-F)/f$ is not increasing, then all players bidding the reserve price $r$ is an optimal collusion. In this mechanism, the seller is used as a correlating device. Under uncertainty on the number of players and false-name bids (Sybils) this mechanism is clearly not Sybil-proof, since players (with $v_i> r$) can send more than one bid with reserve price and unilaterally increase their chances to win the auction. And so, we will focus on twisting the second mechanism provided by the authors. In this mechanism, the authors assume that the cartel is stronger than the one provided before. First, the cartel is able to exclude non-serious bidders. Second, the players are able to make transfers among their members. In this scenario, the authors are able to provide an incentive-compatible (truthful bidding) and efficient (bidding ring) mechanism. The mechanism works as follows. Before the auction, the cartel members report their valuations to the mechanism. If no report exceeds $r$, the cartel does not bid in the auction. If at least one player $i$ exceeds the bid $r$, the bidder making the highest report $v$ obtains the item and pays a total of 
\begin{equation}\label{eq:T}
    T(v) = (n-1)F(v)^{-n}\int^v_r(x-r)F(x)^{n-1}f(x)dx+r.
\end{equation}
Each loser bidder receives from the winner $(T(v)-r)/(n-1)$ and the seller receives $r$. Moreover, they prove that every incentive-compatible and efficient mechanism has the property that the winner transfers to each loser an amount equal to $V(n)=\mathbb E[v_{(2)}-r\mid v_{(1)}]/n$ (where $v_{j}$ is the $j$th order statistic). And so, is not Sybil-proof. The proof of this is straightforward. 
For $n V(n)\uparrow v_h$ and so, for $n$ sufficiently large, holds $2V(n+1)\geq V(n)$. And so, with a sufficient number of players, any bidder $i$ with valuation $v_i\geq r$ has incentives to at least bid twice (one truthfully and another one $r$). In particular, any efficient bidding ring(the bidder with the highest valuation wins the auction and all the surplus generated in the collusion is divided among the bidding ring members) is not Sybil-proof. 
To find permissionless bidding rings that are Sybil-proof and incentive compatible, first, we define the following family of bidding rings defined by $(T,g)$:
\begin{itemize}
    \item Registration phase: All players register an identity responsible for bidding $i=1,...k$ (potentially some identities will be Sybils).
    \item All players submit their bids $w_1,...,w_k$ to the ring center (potentially some players will submit more than one Sybil bid).
    \item The bidder with the highest bid pays $T(v;k)$ to the ring center and submits the highest value to the seller.
    \item W.l.o.g assumes that the winning bid is $i=1$. Then, the ring center pays to each loss identity $i=2,...,k$ the amount $g(k)T(v;k)$ and credible burns the remaning part.
\end{itemize}
Since we want the mechanism to be budget balance, from now on we will assume that $g\leq 1/(k-1)$. Observe that if $g(k)=1/(k-1)$, and $T$ defines as \ref{eq:T}, then we have the previous mechanism. We denote by $\mathcal B$ the set of bidding rings of this family such is incentive compatible to report the bid truthfully.
In the following proposition, we give in terms of $g$, the function $T$ in order to be incentive compatible to report the true valuation. For simplicity, from now on, we will assume that the reserve price $r$ equals zero.
\begin{proposition}\label{prop:bidding} For a given $g$, the bidding ring $(T,g)$ with:
\begin{equation}
    T(v)= F(v)^{-n-l(n)+1}\int^{v}_r(n-1)vF(u)^{n-2+l(n)}f(u)du
\end{equation}
with $l(n)=(n-1)g(n)$ is incentive compatible.
\end{proposition}
Observe that if $g=0$, then $T(x)=\mathbb E[v_{(2)}\mid v_{(1)}=x]$ and so the expected profit of the Bidding ring $(T,g)$ is exactly the expected profit without the bidding ring. Also, as we have seen, for not all $g$, the bidding ring $(T,g)$ is Sybil-proof. 
Let $\mathcal B\mathcal S\subseteq \mathcal B$ be the set of incentive-compatible bidding rings that are Sybil-proof in equilibrium (all reporting exactly one identity is an equilibrium). Observe that $\mathcal B\mathcal S$ is non-empty since for $g=0$, the corresponding bidding ring is Sybil-proof. Therefore, we can consider the optimal Sybil-proof bidding ring $g_{max}$, and the optimal social welfare OPT:
\begin{align*}
g_{max}&=\underset{g; (T,g)\in \mathcal B\mathcal S}{\text{argmax}}\mathbb E \left[\sum_{i=1}^nU_i(w_i,n_i,w_{-i},n_{-i})\mid w_i=v_i,n_i=1,\forall i\right]\\
    \text{OPT}&=\underset{g; (T,g)\in \mathcal B\mathcal S}{\text{max}}\mathbb E\left[\sum_{i=1}^nU_i(w_i,n_i,w_{-i},n_{-i})\mid w_i=v_i,n_i=1,\forall i\right]
\end{align*}
with $U_i$ being the utility function of a player.

\begin{proposition}\label{prop:optimalbr} If the valuations are non-trivial i.e. $\mathbb E[v_{(1)}]\not=0$, then $\text{OPT}>\mathbb E[v_{(1)}-v_{(2)}]$. In other words, there are profitable Sybil-proof bidding rings in second-price auctions.
\end{proposition}

\section{Cost sharing mechanisms}

This section focuses on studying mechanisms proposed in \cite{dobzinski2008shapley} within the context of cost-sharing mechanism design. We specifically explore the susceptibility of these mechanisms to false-name strategies. In particular, we will prove that the mechanisms proposed in \cite{dobzinski2008shapley} are not Sybil-proof.

In a cost-sharing mechanism design problem, several participants with unknown preferences vie to receive some good or service, and each possible outcome has a known cost. Formally, we consider problems defined by a set of players $N$ and a cost function $C:2^N\rightarrow\mathbb R^+$ that describes the cost incurred by the mechanism as a function of the outcome(i.e., of the set $S$ of winners). 

This section focuses on the study of the \textit{public excludable good} problem, which involves determining whether to finance a public good and, if so, identifying who is permitted to use it. We will assume that $C(S)=1$ for every $S\not=\emptyset$ and $C(\emptyset)=0$. A mechanism will consist of an allocation rule $\textbf{x}$ and a payment rule $\textbf{p}$ that will determine which set $S$ is allocated the public good and how much each player must pay. Player $i$ has a private value $v_i$ for being included in the chosen set (having access to the public good). We assume that players have quasi-linear utilities, meaning that each player $i$ aims to maximize $u_i(S,p_i)=v_ix_i-p_i$ where $x_i=1$ if $i\in S$ and $x_i=0$ if $i\not\in S$.

Under these assumptions, in \cite{dobzinski2008shapley} they proposed three different truthful mechanisms, the VCG mechanism, the Shapley mechanism and the Hybrid mechanism. The first one is efficient (welfare maximizer) however, in general, has deficit (i.e. the users payments to not covers the costs incurred by financing the public good). The second one has no deficit, however has $0$ worst-case welfare. On the other hand, the Hybrid mechanism has no-deficit and is $\mathcal H_n$-approximate. Moreover, the authors prove that, up to a constant, this mechanism is tight.
However, as we will see, all these mechanisms are not Sybil-proof (we will show that the Hybrid mechanism is not Sybil-proof, the others follow similarly).

\begin{proposition}\label{prop:cost} The Hybrid mechanism, the Shapley mechanism and the VCG for public excludable goods are not Sybil-proof.
\end{proposition}
The last proposition opens the following question. What is the maximum $\alpha(n)$ such that there is a $\alpha(n)-$approximated truthful, no-deficit, and Sybil-proof cost-sharing mechanism?

%2nd part of the paper: Sybil Commitment.
\section{Sybil Proofness with Commitments}
In the previous section, we considered Sybil games as an extension of normal games. Now, we will explore Sybil commitments, which offer valuable insights into strategic decision-making when agents credibly delegate decisions to third parties, such as AI or reinforcement learning algorithms, for optimizing preferences in complex environments. This delegation is related to the revelation principle, which states that agents can truthfully reveal their private information to a mechanism while maximizing their utility.
Sybil commitment games are an extension of normal games with two phases. In the first phase, each player commits a finite number of Sybils that will act as individual rational agents in the next phase. In the second phase, each Sybil acts under the rationality committed in the previous round and with public knowledge of the number of players in the games. These Sybils are indistinguishable from actual players and are therefore treated as such by other Sybils.
As the Sybil extension game has the structure of a two-phase game, we can analyze its equilibrium points, which we will call Sybil commitment equilibrium. These games have natural applications in environments where the players can credibly commit to act in a certain way up to some predicates. For example, in \cite{hall2021game}, the authors propose a model for games in which the players have shared access to a blockchain that allows them to deploy smart contracts to act on their behalf. Also, in the second phase, Sybils do not need to constrain their actions based on their belief since they have a perfect signal on the number of players.
However, we will prove that in general, it is not incentive-compatible to commit a unique agent or, in other words, reveal the agents' identity. In general, when players play a game with incomplete information on the number of agents, they do not have incentives to truthfully reveal their identity. Finally, we will give a necessary condition for players to truthfully report their identity in the commitment phase.
The study of these games is crucial for designing mechanisms where agents credibly delegate their preferences and decision-making to third parties or to an AI, ultimately promoting trust, accountability, and efficiency in multi-agent settings. Understanding Sybil commitment games helps develop robust systems across various contexts, from financial markets and online platforms to distributed systems and influence campaigns.
%TODO name the map

We will define the Sybil-commitment game of players with utility functions $U_1,...,U_n:\textbf{A}^\infty\times\textbf{A}^\infty\rightarrow\mathbb R$ and $U_i(x,\sigma(y))=U_i(x,y)$ for all $y\in S_{\infty}$\footnote{In general in literature is said that the utility function is anonymous.}. For every normal game $G$, we denote by $\text{NE}(G)$ the set of Nash equilibrium of the game $G$. Then, for a map of this form, we consider the following two-stage game:

\begin{enumerate}
    \item For every player $i$ the set $\mathcal N$ and $U_{-i}$ is unknown. 
    \item In the first phase, the \textit{Sybil-commitment-stage}, each player $i$ chooses and commits a number of Sybils $n_i$ with Sybil set $\mathcal I_i$ with utility functions $\{U^1_i,...,U^j_i\}$. The set of utility functions lay in $\mathcal L\subseteq C(\mathbf{A}\times\mathbf{A}^\infty,\mathbb R)$\footnote{The set of continuous functions from $\mathbf{A}\times\mathbf{A}^\infty$ to $\mathbb R$.}. In other words, the set of preferences that the players can declare lay in $\mathcal L$. The cost of generating Sybils is modelled by a function $C:\mathbb R_+\times \mathbb R_+\rightarrow\mathbb R$ where $C(x,y)$ is the cost of a player that generates $x$ identities and the other players generate $y$ identities. 
    \item In the second phase, all players know the total number $n=\sum_{i=1}^Nn_i$ of players and the set of utility functions $\{U^j_i:\textbf{A}\times\textbf{A}^\infty\rightarrow\mathbb{R}:i\in\mathcal N,j\in \mathcal I_i\}$. In this phase, the Sybils play a normal-form game (complete information) with the number of players and utility specified.
    \item Finally, if the set of actions taken by all Sybils is $a=(a^1,...,a^N)$ the total payoff of the player $i$ is
    \begin{equation*}
         \tilde{U}_i(a^i,a^{-i}) = U_i(a^i,a^{-i})-C(|\mathcal I_i|,|\bigcup_{k\not=i}\mathcal I_k|)
    \end{equation*}
\end{enumerate}

We call this game the \textbf{Sybil commitment extension game}  of the tuple $G=(U_1,...,U_n,C)$ and denoted by $\textbf{SyC}(\mathbf{G},\mathcal L)$.

\textbf{Sybil-commitment equilibrium}: For a Sybil commitment extension game $\textbf{SyC}(\mathbf{G},\mathcal{L})$ we say that a point $(\textbf{U}=\{U^j_i:\textbf{A}\times\textbf{A}^\infty\rightarrow\mathbb{R}:i\in\mathcal N,j\in \mathcal I_i\},\textbf{x})\in (\mathcal L^\infty)^N\times\textbf{A}^\infty$ is a Sybil Nash equilibrium if:
\begin{enumerate}
    \item $\textbf{x}$ is a Nash equilibrium of the normal game $G_n=([n],\textbf{A}^n,U)$ with $\textbf{n}$.
     \item Let $\textbf{U}_i=\{U^j_i:j\in \mathcal I_i\}\subseteq\textbf{U}$. 
     For all $\textbf{U'}_i\not=\textbf{U}_i$ we have that $ \tilde{U}_i(\textbf{x})\geq \max_{\textbf{y}\in \text{NE}(\textbf{U'})}\tilde{U}_i(\textbf{y})$.
\end{enumerate}
And so, we say that a game is \textit{Sybil-commitment-proof (SPC)} if for every number of players, reporting only one identity with truthful utility reporting is a dominant strategy. In other words, providing that information to other players is incentive compatible.
% General lower bound theorem for Sybil Commitment games

\subsection{Lower bound for Sybil Commitment games}
First, we will begin by considering the requirements that must be satisfied for a Sybil game to be a Sybil-commitment-proof game with linear costs. We will then use this theorem to demonstrate that many games that appear in the literature are not SCP. This will be accomplished by computing the price of anarchy \cite{roughgarden2015intrinsic} of these games and by showing that these games fail to meet the necessary conditions for being truthful self-reporting games with linear costs, as outlined by the aforementioned theorem.
\begin{theorem}\label{theorem:main} Let $\textbf{SyC}(\mathbf{G},\mathcal L)$ be Sybil commitment extension game of $(U_1,...,U_n,C(x,y)=c\cdot x)$  with bounded utility functions $U_i$ and $U_i=U_j$ for all $i,j$ (symmetric) and $U_i(x,y)$ is super-additive in the first term. Assume that $c=\mathcal O(1/l^2)$. Then, if $\textbf{SyC}(\mathbf{G},\mathcal L)$ is SPC, then we have that the welfare in equilibrium $\textbf{W}$ of the underlying game is $\mathcal O(n/2^n)$. 
\end{theorem}

As a corollary, if the cost of creating Sybils is sufficiently small and the Welfare $\textbf{W}$ in equilibrium is $o(n/2^n)$, then the game is not SPWC. Also, taking the reward distribution mechanism, we have that this bound is tight.
In the following, we will provide a set of popular games that are not truthful self-reporting. But first, we will give two examples of truthful self-reporting games. In the following we assume that $\mathcal L$ is the set of one function, the actual payoff function of the players.
\begin{itemize}
    \item The trivial game defined by $U:\mathbb R_+^\infty\rightarrow\mathbb R_+^\infty$ with cost function $C(x)=cx$ defined as $U_i(x)=\frac{3c}{2}$ for all $i\in\mathbb N$ and $x\in \mathbb R^\infty$. This game is obviously truthful self-reporting since no action can modify the outcome and reporting two identities has  payoff $-c/2$ and reporting one identity has payoff $c/2$.
    \item The game defined by $U_i(x_i,x_{-i})=\frac{2x_i}{e^{x_i+x_{-i}}}$ and cost function $C(l,k)=\frac{1}{e^{l+k}}$. Observe that $x=1$ is strictly dominant strategy. Since $x=1$ is the unique critical point of $U$ and $U$ is increasing in $0\leq x\leq1$. If there are $k$ reported identities, then $\tilde{U}(l,k)=\frac{l}{e^{l+k}}$ and this is maximized with $l=1$. Therefore, the game is truthful self-reporting.
\end{itemize}
% Examples that are Sybil resistant but not SPWC (Sybil proof with commitments).
\subsection{Non-SYWC examples}
Finally, we will see that not all Sybil-proof games are Sybil-commitment proof games. Consider a homogeneous oligopoly with $n$ firms, where the Cournot equilibrium is regular and unique. The inverse demand function is given by $P(x^T\cdot \textbf{1})$. Each firm (player) chooses an output $x_i$ and its costs are given by $C_i(x_i)$. For an output vector $x=(x_1,...,x_n)$, the profit to each firm is:
\begin{align*}
    U_i(x_i,x_{-i})= x_iP(x^T\cdot \textbf{1})-C_i(x_i)
\end{align*}
Under sufficient good conditions \cite{dastidar2000unique}, this game has a unique locally stable Nash equilibrium. Moreover, if the firms have the same cost function $C$, this equilibrium is symmetric. Now, let's consider the Cournot oligopoly game with $p(x)=\alpha-x$ and $C(x)=cx$ for some $\alpha,c\in\mathbb R_{\geq0}$ such that $c<\alpha$. So in this game, we have that 
\begin{equation*}
    U_i(x_i,x_{-i})=x_i(\beta-x^T\cdot \textbf{1})
\end{equation*}
where $\beta = \alpha - c$. In this scenario, the unique equilibrium point is given by $q^\star(n)=\frac{\beta}{n+1}$. The payoff of a player $U_i(q^\star)=\frac{\beta^2}{(n+1)^2}$. And so, the welfare is $\Theta(1/n)$. Now, let's take the Sybil Commitment extension of this game. We know by theorem \ref{theorem:main} that if $c=\mathcal O(1/n^2)$, game is not Sybil-commitment-proof with $\mathcal L=\{U_1\}$ (observe all $U_i$ are the same up to reorder of $x_i,x_{-i}$). 
More specifically, since the utility of a firm in equilibrium is $ U_i(q^\star)=\frac{\beta^2}{(n+1)^2}$, if the total number of players is $x_{-i}$, the $i$th player the best response strategy is to generate $x$ sybils such that 
\begin{equation*}
\begin{aligned}
& \underset{x}{\text{maximize }}x\frac{\beta^2}{(x+x_{-i}+1)^2} -cx\\
& \text{subject to } x\in\mathbb Z_{\geq0}.
\end{aligned}
\end{equation*}
For example, if there are $n=10$ players that commit in the Sybil-commitment phase, $\beta=10$ and $c=0.001$, then the best-response of a player is $x=11$.

 More general games such as concave pro-rata games with differential function $f$. We know that these games are Sybil-proof, and  by \cite{johnson2022concave} that concave pro-rata games have welfare in equilibrium $\Theta(1/n)$. Therefore, if $c=\mathcal O(1/n^2)$, we have that this type of game are not SCP.

As we have shown previously, any concave pro-rata game is not Sybil resistant. In the following, we plot the strategy of committing more than one player to the aggregated arbitrage game in decentralized Batch exchanges. We explore through simulation the payoff of generating one Sybil in the aggregated arbitrage game provided in \cite{johnson2022concave}. The aggregated arbitrage game consists of game played by arbitrageurs that can exploit the price difference between a Constant function market maker $C$ with two assets $A$ and $B$ (for more details, see \cite{angeris2022constant}) and an off-chain market maker essentially risk-free. Moreover, the trades in the constant function market maker are batched before they are executed. Specifically, the trades are aggregated in some way (depending on the type of batching performed) and then traded ‘all together’ through the CFMM, before being disaggregated and passed back to the users. In this scenario, the authors assume that the forward function $g$ is differentiable (if a player offers $\Delta$ assets to the CFMM obtains $g(\Delta)$ of $B$ assets). 
Also, they assume that the price of the external market maker is $c$. Therefore, the profit of an arbitrageur trading $t$ assets is $g(t)/c-t$. And so, the optimal arbitrage problem consists of maximizing $g(t)-ct$ constraint to $t\geq0$.

When the trades of the players are batches, arbitrageurs cannot directly trade with the CFMM, but must instead go through the batching process. Assuming that there are $n$ arbitrageurs, the pro-rata game induced is 
\begin{equation*}
    U_i(x_i,x_{-i}) =\frac{x_i}{x_i+x_{-i}}g(x_i+x_{-i})-cx_i
\end{equation*}

As shown in \cite{johnson2022concave} this game is a pro-rata concave game and has a unique Nash equilibria. Fixing the CFMM to be a constant product market maker \cite{angeris2022constant} and fixing the reserves, we can compute the payoff of each arbitrageur in equilibrium. In the case that a player can commit more than one identity, we show that this player has incentives to, at least, commit one more identity if the payoff of each player is the payoff in equilibrium with $n+1$ players, and the payoff of the Sybil attacker is twice that amount, see figure in appendix \ref{fig:arb}\footnote{Code available in \url{https://github.com/BrunoMazorra/CostsOfSybils}}.

An interpretation of this is that if a player can credibly convince the other players that there are two independent players (one himself, and the other a Sybil identity), then he has incentives to do so even if he has to commit to acting as individual rational agents. Clearly, this will decrease the social welfare, since the social welfare is $\mathcal O(1/n)$.

%%
%% Bibliography
%%

%% Please use bibtex, 

\appendix
\section{Appendix}\label{appendix}

\subsection{Figure}

\begin{figure}[!h]
    \centering
    \includegraphics[scale=0.6]{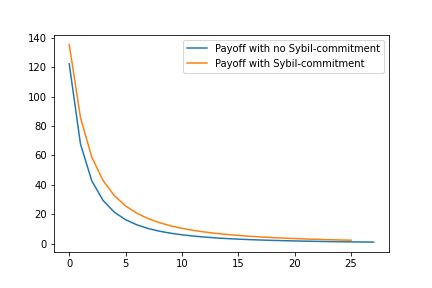}
    \caption{Aggregated arbitrage game with Sybil commitment strategy}
    \label{fig:arb}
\end{figure}

\subsection{Introduction to Topological monoids}

A \textbf{monoid} is a set $M$ together with a binary operation $\cdot$ that satisfies the following axioms:
\begin{enumerate}
    \item \textit{Associativity}: For all $a, b, c \in M$, $(a \cdot b) \cdot c = a \cdot (b \cdot c)$.
    \item \textit{Identity element}: There exists an element $e \in M$ such that $a \cdot e = e \cdot a = a$ for all $a \in M$.
\end{enumerate}
A monoid is commutative if  for every pair of elements $a,b\in G$ it holds $a\cdot b= b\cdot a$.

A \textbf{topological space} is a pair $(X, \mathcal{T})$, where $X$ is a set and $\mathcal{T}$ is a collection of subsets of $X$, called \textit{open sets}, satisfying the following axioms:
\begin{enumerate}
    \item The empty set $\emptyset$ and the whole set $X$ are in $\mathcal{T}$.
    \item The intersection of any finite number of sets in $\mathcal{T}$ is also in $\mathcal{T}$.
    \item The union of any collection of sets in $\mathcal{T}$ is also in $\mathcal{T}$.
\end{enumerate}

A \textbf{topological monoid} is a monoid $(M, \cdot)$ equipped with a topology $\mathcal{T}$ such that the binary operation $\cdot: M \times M \rightarrow M$ are continuous with respect to the topology $\mathcal{T}$. Examples of topological monoids are $(\mathbb R_{>0},1)$ and $(\mathbb R,0)$.
\subsection{Introduction to direct sums}

Direct sums are a way to combine mathematical structures, such as groups, vector spaces, and modules, while preserving their individual properties. In this section, we provide a brief introduction to the concept of direct sums.

Given two groups $(G_1, \cdot_1)$ and $(G_2, \cdot_2)$, their \textbf{direct sum}, denoted by $(G_1 \oplus G_2, \cdot)$, is a group with the underlying set $G_1 \times G_2$ and the binary operation $\cdot$ defined as:
\[(g_1, g_2) \cdot (h_1, h_2) = (g_1 \cdot_1 h_1, g_2 \cdot_2 h_2),\]
for all $g_1, h_1 \in G_1$ and $g_2, h_2 \in G_2$.

Similarly, for two vector spaces $V_1$ and $V_2$ over the same field $F$, their \textbf{direct sum}, denoted by $V_1 \oplus V_2$, is a vector space with the underlying set $V_1 \times V_2$ and the operations of vector addition and scalar multiplication defined component-wise.

The direct sum can be generalized to an arbitrary number of structures. In our case for a family of monoid $\{A_i\}_{i \in I}$, their \textbf{direct sum} is
\begin{equation*}
   \bigoplus_{i \in I} A_i = \{(v_i)_i\in\prod_{i=1}^\infty A_i:v_i=0\text{ for all but finitely many } i\}, 
\end{equation*}
is a vector space with component-wise vector addition and scalar multiplication.

\subsection{Proofs}
\textbf{Proof \ref{prop:game1}}: First, assume that players can take actions in $x\in\mathbb R_+$. Then, the vector
\begin{equation*}
    (h_1,...,h_n)=\frac{n-1}{n^2}R\cdot \textbf{1}^T.
\end{equation*}
 is a Nash equilibrium.  With vector of payoffs $(U_1,...,U_n)=\frac{R}{n^2}\cdot\textbf{1}^T$. M

Assume that other players are consuming $y$ computation resources, since the player $i$ is individually rational, solves 
\begin{equation*}
\begin{aligned}
& \underset{x}{\text{maximize }} U_i(x,y)\\
& \text{subject to } x\geq0
\end{aligned}
\end{equation*}
Clearly $U$ is twice differentiable, and we have that 
\begin{align}
    \frac{\partial U(x,y)}{\partial x} &= \frac{y}{(x+y)^2}R-1\label{equatio_max}\\
    \frac{\partial^2 U(x,y)}{\partial x^2} &= -\frac{y}{(x+y)^3}R<0 \label{equation_concave}.
\end{align}
By \ref{equation_concave} we have that $U$ is concave, and therefore the local maximum are global. By \ref{equatio_max}, we have that the maximum is realized in $x^\star(y)=\sqrt{Ry}-y$. Now, assume that $(x_1,...,x_n)$ is a Nash equilibrium. So, we have that $x_i=\sqrt{Rx_{-i}}-x_{-i}$ for all $i$. Then $x^2_T=Rx_{-i}$ for all $i$ and therefore we have that $x_i=x_j$ for all $i,j$. Implying we have a unique Nash equilibrium. Now, let's compute it. Let $x:=x_1$, we have that $n^2x_1^2=R(n-1)x_1$, so $x_1=\frac{R(n-1)}{n^2}$. The payoff of this equilibrium is $\frac{R}{n^2}$ and so the social welfare in equilibrium is $R/n$. Now, there is a $p$ such that if all players play $\pi$, then $U_i( \lfloor \frac{R/c(n-1)}{n^2}\rfloor,\pi_{-i}) =U_i( \lfloor \frac{R/c(n-1)}{n^2}\rfloor+1,\pi_{-i})$. And so, we deduce that with this $p$, $(\pi,...,\pi)$ is a Nash equilibrium.

\textbf{Proof \ref{prop:optimalproof}}: First, if a mechanism is Sybil-proof for all $c\geq0$ is in particular Sybil-proof for $c=0$. As mentioned previously, we have that for all $y\geq1$, $r(1+y)/(1+y)\geq 2r(2+y)/(2+y)$. And so, we have that
\begin{equation*}
    r(1+n)\leq \frac{r(n)}{2}\frac{n+1}{n},\text{ for }n\geq 1.
\end{equation*}
And so, recursively we deduce that 
\begin{equation*}
 r(1+n)\leq \frac{r(1)}{2^{n}}\prod_{k=1}^{n}\frac{k+1}{k} = \frac{r(1)}{2^{n}}(n+1)
\end{equation*}
Since $r(1)\leq R$, we deduce that $r(n)\leq \frac{n}{2^{n-1}}R=r_{max}$.

\textbf{Proof \ref{lemma}}: Since the mechanism is Sybil-proof, no player have incentives to generate more than one identity. By symmetry and strategy-proofness, we have that exists $s$ such that for all $i$ and $y$, it holds $s = \text{argmax}_{x\in \textbf{A}}U_i(x,y)$. And so, in equilibrium all players play $s$ and the payoff in equilibrium is $U_i(\pi,\pi_{-i})$ with $\pi$ being the $n-$tuple $(s,...,s)$. If we define $r(x)=U_i(\pi,\pi_{-i})$, it holds that the payoff of a player reporting $k$ identities can be reformulated as a discrete pro-rata mechanism.

\textbf{Proof \ref{prop:nearoptimal}}: Consider the point $p_\varepsilon = (K-\varepsilon,1)$. Now, consider the function 
\begin{equation*}
    f_\varepsilon(x) = \begin{cases}Rx/(K-\varepsilon),\text{ if }x\leq K-\varepsilon,\\
                                    -Rx/\varepsilon+RK/\varepsilon,\text{ otherwise}
                        \end{cases}
\end{equation*}
for $\varepsilon>0$. This function is clearly concave and so has a unique Nash equilibrium \cite{johnson2022concave}. Since this function is picewise differential and concave, we can approximate this function by strictly concave functions $\{g_k\}$ such that $g_k\rightarrow f$ and $g_k'(x)\rightarrow f'(x)$ for all $x\not=p$. Since the associated Nash equilibrium is solution of the optimization problem $\max q^nf(q)$, we have that $q_k = \text{argmax} q^ng_k(q)$ converges to $q$ the solution of the optimization problem that makes the unique Nash equilibrium. On the other hand, since $g_k$ is is strictly convex, we have that it holds
\begin{equation*}
    -q_kg_k'(q_k)/(n-1) = g_k(q_k)
\end{equation*}
So, making $k\rightarrow+\infty$, we have that $f(q) = -qf'(q)/(n-1)$. Computing $q$, we deduce that $f(q)\rightarrow R$ when $\varepsilon\rightarrow0$. Since $f(q)$ is the welfare in equilibrium, we have proved the proposition.

\textbf{Proof \ref{prop:cake2}} Observe that the expected size of the slice of player $i$ is 
\begin{equation*}
    \Pr[X=1]\left (\sum_j\mu_i(C_j)\Pr[\sigma(i)=j]\right ) = \frac{n}{2^{n-1}}\sum_j\mu_i(C_j)/k=\frac{n}{2^{n-1}}\mu_i(\cup_j C_j)=\frac{1}{2^{n-1}}.
\end{equation*}
 Since this quantity is independent of $\nu_i$ then player $i$ has no incentive to declare $\nu_i\not=\mu_i$. And so, the mechanism is truthful and expectancy $\frac{1}{2^{n-1}}$-proportional. Now, if the player declares $k$ identities and other players report $y$. We write $n=y+k$, then the expected payoff is:
\begin{align*}
    \sum_{\mathcal I\subseteq [n]:|\mathcal I|=k}\left (\sum_{j\in\mathcal I}\mu(C_j)\Pr[X=1]\right)\Pr[\sigma([k])=\mathcal I]&=\Pr[X=1]\Pr[\sigma([k])=[k])\sum_{j\in[n]}\sum_{\mathcal I\subseteq [n]:|\mathcal I|=k,j\in\mathcal I}\mu_i(C_j)\\
    &=\frac{n}{2^{n-1}}\frac{y!k!}{n!}\sum_{j\in[n]}\mu_i(C_j){n-1 \choose k-1}\\
    &=\frac{k}{2^{y+k-1}}
\end{align*}
Since the function $f(x) = \frac{x}{2^{y+x-1}}$ restricted to $\mathbb Z_{\geq0}$ is maximized in $x=1$, we have that the mechanism is Sybil-proof.

\textbf{Proof \ref{prop:bidding}}: First, we define $T(v)$ as the amount being payed by the player with highest valuation. And we define $g:\mathbb N\rightarrow \mathbb [0,1]$ to be the fraction of $T(v)$ being paid to a reported player, that is losing player will obtain will get paid $g(n)(T(v)-r)$. Since the protocol is budget balance, we have that $g(n)\leq 1/(n-1)$. Now, let $\pi(w,v,m)$ be the expected payoff of a player with valuation $v$, report $w$ and reporting $m$ identities to the protocol. From now on we will denote by $n$ the number of actual players (not the reported ones) that are members of the cartel. In this scenario, it holds:
\begin{align*}
    \pi(w,v,1) &= (v-T(w))F(w)^{n-1}+\\&(1-F(w)^{n-1})\int^{v_h}_w g(m+n-1)(T(u)-r)\frac{(n-1)F(u)^{n-2}f(u)}{1-F(w)^{n-1}}du\\
    &=(v-T(w))F(w)^{n-1}+\int^{v_h}_wg(n+m-1)(n-1)(T(u)-r)F(u)^{n-2}f(u)du
\end{align*}
And so, if we denote by $l(n)=(n-1)g(n)$, we have
\begin{align*}
    \frac{\partial \pi}{\partial w} &=(n-1)(v-T(w))F(w)^{n-2}f(w)-T'(w)F(w)^{n-1}-l(n)F(w)^{n-2}f(w)(T(w)-r)\\
    &=\left ( (n-1)v - (n-1+l(n))T(w)+l(n)r \right)F(w)^{n-2}f(w)-T'(w)F(w)^{n-1}
\end{align*}
Since $\partial^2 \pi/\partial w\partial v\geq0$, incentive compatibility is characterized by 
\begin{equation*}
    \frac{\partial \pi}{\partial w}\mid_{w=v}=0
\end{equation*}
This induces the following differential equation:
\begin{equation}
    ((n-1)v)F(v)^{n-2}f(v) =((n-1)+l(n))T(v)f(v)F(v)^{n-2}+T'(v)F(v)^{n-1}
\end{equation}
Multiplying by $F(v)^{l(n)}$ we obtain
\begin{equation}
    ((n-1)v)F(v)^{n-2+l(n)}f(v) =(n-1+l(n))T(v)f(v)F(v)^{n-2+l(n)}+T'(v)F(v)^{n+l(n)-1}
\end{equation}
Integrating in both sides, we obtain
\begin{align}
    \int^{v}_r((n-1)v)F(u)^{n-2+l(n)}f(u)du &= T(v)F(v)^{n+l(n)-1}\\
    F(v)^{-n-l(n)+1}\int^{v}_r((n-1)v)F(u)^{n-2+l(n)}f(u)du &= T(v) \square
\end{align}

\textbf{Proof }\ref{prop:cake1}: Assume all players have the same valuation $\mu$. Also, we assume that there is a Sybil-resistant cake-cutting algorithm that is $\alpha$-proportional. Then, if players report  $\mu$ each player obtains some slice with exactly the same valuation $\beta(n) \geq \alpha$ (by symmetry). Since is Sybil-proof, it holds $\beta(n-1+k)k\leq \beta(n)$ for all $k\geq1$. In particular, $2\beta(n+1)\leq \beta(n)$. And so, we deduce that $\beta(n)\leq 1/2^{n-1}$, then $\alpha\leq 1/2^{n-1}$. The worst case welfare is in particular smaller than $n\beta(n)\leq  n/2^{n-1}$. $\square$

\textbf{Proof }\ref{prop:optimalbr} This proposition can be deduced by imposing that the expected profits has maximum profits with $m=1$. If we construct $g$ to be twice differential, this induces a differential equation with a unique solution. Since $g$ is non-trivial, the optimal strategy has strictly larger welfare than the one with $g=0$ (the second price auction).$\square$

\textbf{Cost sharing mechanism}
\begin{mybox2}{Hybrid mechanism}
\begin{enumerate}
    \item Accept a bid $b_i$ from each player $i$.
    \item Let 
    \begin{equation*}
        S^\star\in\text{argmax}_{S\subseteq N}\left\{\sum_{i\in N}b_i-C(S)\right\}
    \end{equation*}
    denote a welfare-maximizing outcome.
    \item Initialize $S=S^\star$.
    \item If $b_i\geq C(S^\star)/|S|$ for every $i\in S$, then halt with winners $S$.
    \item Let $i^\star$ be a player with $b_i<C(S^\star)/|S|$.
    \item Set $S\leftarrow S\setminus\{i\}$ and return to Step 4.
    \item Charge each winner $i\in S$ a payment equal to the minimum bid at which $i$ would continue to win.
\end{enumerate}
\end{mybox2}
\textbf{Proof }\ref{prop:cost} Assume that $C(S)=1$ and $v_1=1+\varepsilon$ and $v_2=v_3=1/3-\varepsilon$. Then the outcome of the mechanism is $S=\{1\}$ with $p_1=1$. On the other hand, the player $1$ splits its bid in two $b_1=1/4$ and $b_1=1/4$, the outcome is $S=\{1,2,3,4\}$ with payment $p\leq 1/2$.

\textbf{Proof }\ref{theorem:main} Let $R=\max_{\textbf{x}\in\mathcal A^\infty}W(\textbf{x})$. Since is Sybil-proof, we have that $R=\max_{\textbf{x}\in\mathcal A^1}W(\textbf{x})$.  By definition of the Sybil extension, we have that 
\begin{equation*}
     U_i(\textbf{x}_i,\textbf{x}_{-i})-c\geq U_i(\textbf{y}_i,\textbf{y}_{-i})-c+ U_j(\textbf{y}_j,\textbf{y}_{-j})-c
\end{equation*}
for all $i\in [n]$, $j\in[n+1]$, $\textbf{x}\in \text{NE}(n)$ and $\textbf{y}\in\text{NE}(n+1)$. By symmetry of $U$, we can assume that $U_{n+1}(\textbf{y}_{n},\textbf{y}_{-(n+1)})\leq W(\textbf{y})/(n+1)$. Fixing $i$, we have that
\begin{align*}
nU_i(\textbf{x}_i,\textbf{x}_{-i})&\geq nU_i(\textbf{y}_i,\textbf{y}_{-i})+\sum_{j\not=i}^{n+1}U_j(\textbf{y}_j,\textbf{y}_{-j})-cn\\
& \geq (n-1)U_i(\textbf{y}_i,\textbf{y}_{-i})+W(\textbf{y})-cn
\end{align*}
and so, adding all $i$, we have that
\begin{align*}
    n W(\textbf{x})&\geq \sum_{i=1}^n[(n-1)U_i(\textbf{y}_i,\textbf{y}_{-i})+W(\textbf{y})]-cn^2\\
    &= (n-1)\left(W(\textbf{y})-U_{n+1}(\textbf{y}_{n},\textbf{y}_{-(n+1)})\right)+nW(\textbf{y})-cn^2\\
    &\geq (n-1)W(\textbf{y})(1-\frac{1}{n+1})+nW(\textbf{y})-cn^2
\end{align*}
and so
\begin{align*}
    W(\textbf{x})&\geq W(\textbf{y})\left ( 1-\frac{1}{n}\right)\left ( 1-\frac{1}{n+1}\right)+W(\textbf{y})-cn\\
    &=2\left ( 1-\frac{1}{n+1}\right)W(\textbf{y})-cn
\end{align*}
By induction, we have that for all equilibrium $\textbf{x}\in NE(n)$,
\begin{align*}
W(\textbf{x})&= \frac{1}{2^{n-2}\prod_{k=2}^{n-1}\left (1-\frac{1}{k}\right)}R+\sum_{k=1}^{n-2}\frac{c(n-k)}{2^k\prod_{l=1}^k\left(1-1/(n-l)\right)}\\
&\leq \frac{1}{2^{n-2}\prod_{k=2}^{n-1}\left (1-\frac{1}{k}\right)}R+\frac{c}{2}n^2=\frac{n-1}{2^{n-2}}R+\frac{c}{2}n^2\sim \frac{n}{2^n}(R/4)\text{ as }n\rightarrow +\infty
\end{align*}
using that $\prod_{k=2}^{n-1}\left (1-\frac{1}{k}\right)=\frac{1}{n-1}$ and $c=\mathcal O(1/l^2)$.

\textbf{Prior-optimal reward distribution mechanism}: Now, let's assume that the distribution $\mathcal D$ is common knowledge. The mechanism designer objective is to find a symmetric mechanism that maximizes the expected welfare.

Let $\mathcal D'$ be the distribution of $\mathcal D$ conditioned to $\mathcal D\geq 1$. We write $p_n = \Pr [\mathcal D'=n]$. Our objective is to find a Sybil proof symmetric mechanism that maximizes the expected social welfare. We want to split some total "reward" $R$ among $n$ players. We can model the distribution with a map $v:\mathbb N\rightarrow \mathbb R$ such that if there are a total number of $n$ players reported, each player receives $v(n)$ (or with abuse of notation $v_n$ by embedding the maps of this form in $\mathbb R^\infty$). Therefore, if there are $n$ identities reported and a player reports $k$ identities, he obtains $k\cdot v_{k+n}$. Our objective is to find $v$ such that:
\begin{enumerate}
    \item (Weak Sybil-proof) All players reporting one identity is an equilibrium (Weaker  than Syibil-proof), formally
    \begin{equation}
    \mathbb E_{n\sim\mathcal D'}[v_{n}]\geq \mathbb E_{n\sim\mathcal D'}[y\cdot v_{y-1+n}]\text{ for all }y\geq 1.
    \end{equation}
    \item (Budget Balance) For every reported number of identities $n$, it holds $n\cdot v_n\leq R$.
    \item (Efficiency) $v$ maximizes the expected welfare of all players in the Sybil-proof equilibrium. That is 
\begin{equation}
v^* = \text{argmax}_{v}\quad\mathbb E_{n\sim\mathcal D'}[n\cdot v_{n}]
\end{equation}
\end{enumerate}

And so, we can rewrite this problem as the optimization problem

\begin{equation*}
\begin{aligned}
& \underset{v}{\text{maximize }}\sum_{n=0}^{+\infty}n\cdot v_n\cdot p_n\\
& \text{subject to } 0\leq v_n\leq R/n\text{ for all }n\geq1\\
& \quad\quad\quad\quad\sum_{n=0}^\infty (v_n-v_{y-1+n}\cdot y)p_{n}\geq0\text{ for all }y\geq2
\end{aligned}
\end{equation*}
If we denote  the welfare by $W_n$ with $n$ reported players, we have that $W_n =n v_n$ and so the optimization problem can be rewritten as:

\begin{equation*}
\begin{aligned}
& \underset{W}{\text{maximize }}\sum_{n=0}^{+\infty}W_n\cdot p_n\\
& \text{subject to } 0\leq W_n\leq 1\text{ for all }n\geq1\\
& \quad\quad\quad\quad\sum_{n=0}^\infty \left(W_n/n-W_{y+n-1}y/(y+n-1)\right)p_{n}\geq0\text{ for all }y\geq2
\end{aligned}
\end{equation*}


\begin{thebibliography}{10}

\bibitem{alkalay2014false}
Colleen Alkalay-Houlihan and Adrian Vetta.
\newblock False-name bidding and economic efficiency in combinatorial auctions.
\newblock In {\em Proceedings of the AAAI Conference on Artificial
  Intelligence}, volume~28, 2014.

\bibitem{angeris2022constant}
Guillermo Angeris, Akshay Agrawal, Alex Evans, Tarun Chitra, and Stephen Boyd.
\newblock Constant function market makers: Multi-asset trades via convex
  optimization.
\newblock In {\em Handbook on Blockchain}, pages 415--444. Springer, 2022.

\bibitem{baza2020detecting}
Mohamed Baza, Mahmoud Nabil, Mohamed~MEA Mahmoud, Niclas Bewermeier, Kemal
  Fidan, Waleed Alasmary, and Mohamed Abdallah.
\newblock Detecting sybil attacks using proofs of work and location in vanets.
\newblock {\em IEEE Transactions on Dependable and Secure Computing},
  19(1):39--53, 2020.

\bibitem{brams1995envy}
Steven~J Brams and Alan~D Taylor.
\newblock An envy-free cake division protocol.
\newblock {\em The American Mathematical Monthly}, 102(1):9--18, 1995.

\bibitem{chen2022sybil}
Hongyin Chen, Xiaotie Deng, Ying Wang, Yue Wu, and Dengji Zhao.
\newblock Sybil-proof diffusion auction in social networks.
\newblock {\em arXiv preprint arXiv:2211.01984}, 2022.

\bibitem{chen2019axiomatic}
Xi~Chen, Christos Papadimitriou, and Tim Roughgarden.
\newblock An axiomatic approach to block rewards.
\newblock In {\em Proceedings of the 1st ACM Conference on Advances in
  Financial Technologies}, pages 124--131, 2019.

\bibitem{cheng2006manipulability}
Alice Cheng and Eric Friedman.
\newblock Manipulability of pagerank under sybil strategies, 2006.

\bibitem{cheng2004notes}
Shih-Fen Cheng, Daniel~M Reeves, Yevgeniy Vorobeychik, and Michael~P Wellman.
\newblock Notes on equilibria in symmetric games.
\newblock 2004.

\bibitem{dastidar2000unique}
Krishnendu~Ghosh Dastidar.
\newblock Is a unique cournot equilibrium locally stable?
\newblock {\em Games and Economic Behavior}, 32(2):206--218, 2000.

\bibitem{dinger2006defending}
Jochen Dinger and Hannes Hartenstein.
\newblock Defending the sybil attack in p2p networks: Taxonomy, challenges, and
  a proposal for self-registration.
\newblock In {\em First International Conference on Availability, Reliability
  and Security (ARES'06)}, pages 8--pp. IEEE, 2006.

\bibitem{dobzinski2008shapley}
Shahar Dobzinski, Aranyak Mehta, Tim Roughgarden, and Mukund Sundararajan.
\newblock Is shapley cost sharing optimal?
\newblock In {\em Algorithmic Game Theory: First International Symposium, SAGT
  2008, Paderborn, Germany, April 30-May 2, 2008. Proceedings 1}, pages
  327--336. Springer, 2008.

\bibitem{douceur2002sybil}
John~R Douceur.
\newblock The sybil attack.
\newblock In {\em International workshop on peer-to-peer systems}, pages
  251--260. Springer, 2002.

\bibitem{fioravanti2022false}
Federico Fioravanti and Jordi Mass{\'o}.
\newblock False-name-proof and strategy-proof voting rules under separable
  preferences.
\newblock {\em Available at SSRN 4175113}, 2022.

\bibitem{hall2021game}
Mathias Hall-Andersen and Nikolaj~I Schwartzbach.
\newblock Game theory on the blockchain: a model for games with smart
  contracts.
\newblock In {\em International Symposium on Algorithmic Game Theory}, pages
  156--170. Springer, 2021.

\bibitem{hashmi2008authentication}
Sarosh Hashmi and John Brooke.
\newblock Authentication mechanisms for mobile ad-hoc networks and resistance
  to sybil attack.
\newblock In {\em 2008 Second International Conference on Emerging Security
  Information, Systems and Technologies}, pages 120--126. IEEE, 2008.

\bibitem{iwasaki2010worst}
Atsushi Iwasaki, Vincent Conitzer, Yoshifusa Omori, Yuko Sakurai, Taiki Todo,
  Mingyu Guo, and Makoto Yokoo.
\newblock Worst-case efficiency ratio in false-name-proof combinatorial auction
  mechanisms.
\newblock In {\em Proceedings of the 9th International Conference on Autonomous
  Agents and Multiagent Systems: volume 1-Volume 1}, pages 633--640, 2010.

\bibitem{jensen2018aggregative}
Martin~Kaae Jensen.
\newblock Aggregative games.
\newblock In {\em Handbook of Game Theory and Industrial Organization, Volume
  I}, pages 66--92. Edward Elgar Publishing, 2018.

\bibitem{johnson2022concave}
Nicholas~AG Johnson, Theo Diamandis, Alex Evans, Henry de~Valence, and
  Guillermo Angeris.
\newblock Concave pro-rata games.
\newblock 2022.

\bibitem{kajii1997robustness}
Atsushi Kajii and Stephen Morris.
\newblock The robustness of equilibria to incomplete information.
\newblock {\em Econometrica: Journal of the Econometric Society}, pages
  1283--1309, 1997.

\bibitem{kamvar2003eigentrust}
Sepandar~D Kamvar, Mario~T Schlosser, and Hector Garcia-Molina.
\newblock The eigentrust algorithm for reputation management in p2p networks.
\newblock In {\em Proceedings of the 12th international conference on World
  Wide Web}, pages 640--651, 2003.

\bibitem{king2012ppcoin}
Sunny King and Scott Nadal.
\newblock Ppcoin: Peer-to-peer crypto-currency with proof-of-stake.
\newblock {\em self-published paper, August}, 19(1), 2012.

\bibitem{krishna2009auction}
Vijay Krishna.
\newblock {\em Auction theory}.
\newblock Academic press, 2009.

\bibitem{levine2006survey}
Brian~Neil Levine, Clay Shields, and N~Boris Margolin.
\newblock A survey of solutions to the sybil attack.
\newblock {\em University of Massachusetts Amherst, Amherst, MA}, 7:224, 2006.

\bibitem{lin2017sybil}
Jian Lin, Ming Li, Dejun Yang, Guoliang Xue, and Jian Tang.
\newblock Sybil-proof incentive mechanisms for crowdsensing.
\newblock In {\em IEEE INFOCOM 2017-IEEE Conference on Computer
  Communications}, pages 1--9. IEEE, 2017.

\bibitem{maram2021candid}
Deepak Maram, Harjasleen Malvai, Fan Zhang, Nerla Jean-Louis, Alexander Frolov,
  Tyler Kell, Tyrone Lobban, Christine Moy, Ari Juels, and Andrew Miller.
\newblock Candid: Can-do decentralized identity with legacy compatibility,
  sybil-resistance, and accountability.
\newblock In {\em 2021 IEEE Symposium on Security and Privacy (SP)}, pages
  1348--1366. IEEE, 2021.

\bibitem{marshall2007bidder}
Robert~C Marshall and Leslie~M Marx.
\newblock Bidder collusion.
\newblock {\em Journal of Economic Theory}, 133(1):374--402, 2007.

\bibitem{marshall2014economics}
Robert~C Marshall and Leslie~M Marx.
\newblock {\em The economics of collusion: Cartels and bidding rings}.
\newblock Mit Press, 2014.

\bibitem{mazorra2022price}
Bruno Mazorra, Michael Reynolds, and Vanesa Daza.
\newblock Price of mev: Towards a game theoretical approach to mev.
\newblock In {\em Proceedings of the 2022 ACM CCS Workshop on Decentralized
  Finance and Security}, pages 15--22, 2022.

\bibitem{mcafee1992bidding}
R~Preston McAfee and John McMillan.
\newblock Bidding rings.
\newblock {\em The American Economic Review}, pages 579--599, 1992.

\bibitem{mossel2010truthful}
Elchanan Mossel and Omer Tamuz.
\newblock Truthful fair division.
\newblock In {\em International Symposium on Algorithmic Game Theory}, pages
  288--299. Springer, 2010.

\bibitem{muller2008sybil}
Wolf M{\"u}ller, Henryk Pl{\"o}tz, Jens-Peter Redlich, and Takashi Shiraki.
\newblock Sybil proof anonymous reputation management.
\newblock In {\em Proceedings of the 4th international conference on Security
  and privacy in communication netowrks}, pages 1--10, 2008.

\bibitem{nakamoto2008bitcoin}
Satoshi Nakamoto.
\newblock Bitcoin: A peer-to-peer electronic cash system.
\newblock {\em Decentralized Business Review}, page 21260, 2008.

\bibitem{robertson1998cake}
Jack Robertson and William Webb.
\newblock {\em Cake-cutting algorithms: Be fair if you can}.
\newblock AK Peters/CRC Press, 1998.

\bibitem{roughgarden2010algorithmic}
Tim Roughgarden.
\newblock Algorithmic game theory.
\newblock {\em Communications of the ACM}, 53(7):78--86, 2010.

\bibitem{roughgarden2015intrinsic}
Tim Roughgarden.
\newblock Intrinsic robustness of the price of anarchy.
\newblock {\em Journal of the ACM (JACM)}, 62(5):1--42, 2015.

\bibitem{sanchez2019zero}
David~Cerezo S{\'a}nchez.
\newblock Zero-knowledge proof-of-identity: Sybil-resistant, anonymous
  authentication on permissionless blockchains and incentive compatible,
  strictly dominant cryptocurrencies.
\newblock {\em arXiv preprint arXiv:1905.09093}, 2019.

\bibitem{sher2012optimal}
Itai Sher.
\newblock Optimal shill bidding in the vcg mechanism.
\newblock {\em Economic Theory}, 50:341--387, 2012.

\bibitem{so2011defending}
Jung~Ki So and Douglas~S Reeves.
\newblock Defending against sybil nodes in bittorrent.
\newblock In {\em International Conference on Research in Networking}, pages
  25--39. Springer, 2011.

\bibitem{suyama2005strategy}
Takayuki Suyama and Makoto Yokoo.
\newblock Strategy/false-name proof protocols for combinatorial multi-attribute
  procurement auction.
\newblock {\em Autonomous Agents and Multi-Agent Systems}, 11:7--21, 2005.

\bibitem{todo2011false}
Taiki Todo, Atsushi Iwasaki, and Makoto Yokoo.
\newblock False-name-proof mechanism design without money.
\newblock In {\em The 10th International Conference on Autonomous Agents and
  Multiagent Systems-Volume 2}, pages 651--658, 2011.

\bibitem{wagman2008optimal}
Liad Wagman and Vincent Conitzer.
\newblock Optimal false-name-proof voting rules with costly voting.
\newblock In {\em AAAI}, volume~8, pages 190--195, 2008.

\bibitem{yokoo2001robust}
Makoto Yokoo, Yuko Sakurai, and Shigeo Matsubara.
\newblock Robust combinatorial auction protocol against false-name bids.
\newblock {\em Artificial Intelligence}, 130(2):167--181, 2001.

\bibitem{yokoo2004effect}
Makoto Yokoo, Yuko Sakurai, and Shigeo Matsubara.
\newblock The effect of false-name bids in combinatorial auctions: New fraud in
  internet auctions.
\newblock {\em Games and Economic Behavior}, 46(1):174--188, 2004.

\bibitem{yu2008sybillimit}
Haifeng Yu, Phillip~B Gibbons, Michael Kaminsky, and Feng Xiao.
\newblock Sybillimit: A near-optimal social network defense against sybil
  attacks.
\newblock In {\em 2008 IEEE Symposium on Security and Privacy (sp 2008)}, pages
  3--17. IEEE, 2008.

\bibitem{yu2006sybilguard}
Haifeng Yu, Michael Kaminsky, Phillip~B Gibbons, and Abraham Flaxman.
\newblock Sybilguard: defending against sybil attacks via social networks.
\newblock In {\em Proceedings of the 2006 conference on Applications,
  technologies, architectures, and protocols for computer communications},
  pages 267--278, 2006.

\bibitem{yu2009dsybil}
Haifeng Yu, Chenwei Shi, Michael Kaminsky, Phillip~B Gibbons, and Feng Xiao.
\newblock Dsybil: Optimal sybil-resistance for recommendation systems.
\newblock In {\em 2009 30th IEEE Symposium on Security and Privacy}, pages
  283--298. IEEE, 2009.

\bibitem{zhang2019double}
Shijie Zhang and Jong-Hyouk Lee.
\newblock Double-spending with a sybil attack in the bitcoin decentralized
  network.
\newblock {\em IEEE transactions on Industrial Informatics}, 15(10):5715--5722,
  2019.

\end{thebibliography}
\end{document}